\documentclass[english]{article}
\usepackage[T1]{fontenc}
\usepackage[latin9]{inputenc}
\usepackage{geometry}
\usepackage{amsmath}
\usepackage{amsthm}
\usepackage{amssymb}
\usepackage{graphicx}

\makeatletter
\numberwithin{equation}{section}
\numberwithin{figure}{section}
\theoremstyle{plain}
\newtheorem{thm}{\protect\theoremname}
\theoremstyle{plain}
\newtheorem{lem}[thm]{\protect\lemmaname}
\theoremstyle{definition}
\newtheorem{defn}[thm]{\protect\definitionname}
\theoremstyle{definition}
\newtheorem{xca}[thm]{\protect\exercisename}
\theoremstyle{plain}
\newtheorem{fact}[thm]{\protect\factname}

\usepackage{hyperref}
\hypersetup{
    colorlinks,
    allcolors=red
}
\usepackage[lined,boxed,ruled,norelsize,algo2e]{algorithm2e}
\@addtoreset{section}{part}

\makeatother

\usepackage{babel}
\providecommand{\definitionname}{Definition}
\providecommand{\exercisename}{Exercise}
\providecommand{\factname}{Fact}
\providecommand{\lemmaname}{Lemma}
\providecommand{\theoremname}{Theorem}

\begin{document}
\global\long\def\defeq{\stackrel{\mathrm{{\scriptscriptstyle def}}}{=}}%
\global\long\def\norm#1{\left\Vert #1\right\Vert }%
\global\long\def\R{\mathbb{R}}%
\global\long\def\Diag{\mathrm{Diag}}%
 
\global\long\def\Rn{\mathbb{R}^{n}}%
\global\long\def\tr{\mathrm{Tr}}%
\global\long\def\diag{\mathrm{diag}}%
\global\long\def\rank{\mathrm{rank}}%
\global\long\def\ma{\mathbf{A}}%
\global\long\def\mx{\mathbf{X}}%
\global\long\def\ms{\mathbf{S}}%
\global\long\def\mproj{\mathbf{P}}%
\global\long\def\mi{\mathbf{I}}%
\global\long\def\mm{\mathbf{M}}%
\global\long\def\mh{\mathbf{H}}%
\global\long\def\mzero{\mathbf{0}}%
\global\long\def\ml{\mathbf{L}}%
\global\long\def\mb{\mathbf{B}}%
\global\long\def\oma{\overline{\ma}}%
\global\long\def\md{\mathbf{D}}%
\global\long\def\mw{\mathbf{W}}%
\global\long\def\omx{\overline{\mathbf{X}}}%
\global\long\def\oms{\overline{\mathbf{S}}}%
\global\long\def\ox{\overline{x}}%
\global\long\def\os{\overline{s}}%

\global\long\def\ttag#1{\tag{#1}}%
\global\long\def\PriR{\mathcal{P}}%
 
\global\long\def\DualR{\mathcal{D}}%
 
\title{Tutorial on the Robust Interior Point Method}
\author{Yin Tat Lee\thanks{University of Washington \& Microsoft Research. Email: yintat@uw.edu.
Supported in part by NSF awards CCF-1749609, DMS-1839116, DMS-2023166,
CCF-2105772, Microsoft Research Faculty Fellowship, Sloan Research
Fellowship and Packard Fellowship.}\\
\\
\and Santosh S. Vempala\thanks{Georgia Tech. Email: vempala@gatech.edu. Supported in part by NSF
awards DMS-1839323, CCF-1909756, CCF-2007443 and CCF-2106444.}\\
\\
}
\maketitle
\begin{abstract}
We give a short, self-contained proof of the interior point method
and its robust version.
\end{abstract}

\section{Introduction}

Consider the primal linear program
\begin{equation}
\min_{\ma x=b,x\in\R_{\geq0}^{n}}c^{\top}x\ttag P\label{eq:tut_primal_LP}
\end{equation}
and its dual
\begin{equation}
\max_{\ma^{\top}y+s=c,s\in\R_{\geq0}^{n}}b^{\top}y\ttag D\label{eq:tut_dual_LP}
\end{equation}
where $\ma\in\R^{d\times n}$ and $\R_{\geq0}=\{x\geq0\}$. The feasible
regions for the two programs are 
\[
\PriR=\{x\in\R_{\geq0}^{n}:\ma x=b\}\mbox{ and }\DualR=\{s\in\R_{\geq0}^{n}:\ma^{\top}y+s=c\text{ for some }y\}.
\]
We define their interiors: 
\[
\PriR^{\circ}=\{x\in\R_{>0}^{n}:\ma x=b\}\mbox{ and \ensuremath{\DualR^{\circ}}}=\{s\in\R_{>0}^{n}:\ma^{\top}y+s=c\text{ for some }y\}.
\]

To motivate the main idea of the interior point method, we recall
the optimality condition for linear programs.
\begin{thm}[Complementary Slackness]
\label{thm:comp_slackness}Any $x\in\PriR$ and $s\in\DualR$ are
optimal if and only if $x^{\top}s=0$. Moreover, if both $\PriR$
and $\DualR$ are non-empty, there exist $x^{*}\in\PriR$ and $s^{*}\in\DualR$
such that $(x^{*})^{\top}s^{*}=0$ and $x^{*}+s^{*}>0$.
\end{thm}

More generally, the quantity $x^{\top}s$ measures the duality gap
of the feasible solution:
\begin{lem}[Duality Gap]
\label{lem:duality_gap}For any $x\in\PriR$ and $s\in\DualR$, the
duality gap $c^{\top}x-b^{\top}y=x^{\top}s$. In particular $c^{\top}x\leq\min_{x\in\PriR}c^{\top}x+x^{\top}s$.
\end{lem}

\begin{proof}
Using $\ma x=b$ and $\ma^{\top}y+s=c$, we can compute the duality
gap as follows
\[
c^{\top}x-b^{\top}y=c^{\top}x-(\ma x)^{\top}y=c^{\top}x-x^{\top}(\ma y)=x^{\top}s.
\]
By weak duality, we have
\[
c^{\top}x=b^{\top}y+x^{\top}s\leq\max_{\ma^{\top}y+s=c,s\in\R_{\geq0}^{n}}b^{\top}y+x^{\top}s\leq\min_{x\in\PriR}c^{\top}x+x^{\top}s.
\]
\end{proof}
The main implication of Lemma \ref{lem:duality_gap} is that any feasible
$(x,s)$ with small $x^{\top}s$ is a nearly optimal solution of the
linear program. This leads us to primal-dual algorithms in which we
start with a feasible primal-dual solution pair $(x,s)$ and iteratively
update the solution to decrease the duality gap $x^{\top}s$. 

\section{Interior Point Method}

In this section, we discuss the classical short-step interior point
method. For two vectors $a,b$, we use $ab$ to denote the vector
with components $(ab)_{i}=a_{i}b_{i}$ and $a/b$ to denote the vector
with components $a_{i}/b_{i}$. For a scalar $t\in\R$, we let $t\mathbf{1}$
denote the vector with all coordinates equal to $t$. 
\begin{defn}[Central Path]
\label{def:central_path}We define the central path $(x_{t},s_{t})\in\PriR^{\circ}\times\DualR^{\circ}$
by $x_{t}s_{t}=t$. We say $x_{t}$ is on the central path of (\ref{eq:tut_primal_LP})
at $t$.
\end{defn}

The algorithm maintains a pair $(x,s)\in\PriR^{\circ}\times\DualR^{\circ}$
and a scalar $t>0$ satisfying the invariant 
\[
\|\frac{xs}{t}-1\|_{2}\leq\frac{1}{4}.
\]
Note that the deviation from the central path is measured in $\ell_{2}$
norm. In each step, the algorithm decreases $t$ by a factor of $1-\Omega(n^{-1/2})$
while maintaining the invariant.

\subsection{Basic Property of a Step}

To see why there is a pair $(x,s)\in\PriR^{\circ}\times\DualR^{\circ}$
satisfying the invariant, we prove the following generalization.
\begin{lem}[Quadrant Representation of Primal-Dual]
\label{lem:primal_dual_representation}Suppose $\PriR$ is non-empty
and bounded. For any positive vector $\mu\in\R_{>0}^{n}$, there is
an unique pair $(x_{\mu},s_{\mu})\in\PriR^{\circ}\times\DualR^{\circ}$
such that $x_{\mu}s_{\mu}=\mu$. Furthermore, $x_{\mu}=\min_{x\in\PriR}f_{\mu}(x)$
where
\[
f_{\mu}(x)=c^{\top}x-\sum_{i=1}^{n}\mu_{i}\ln x_{i}.
\]
\end{lem}

\begin{proof}
Fix $\mu\in\R_{>0}^{n}$. We define $x_{\mu}=\arg\min_{x\in\PriR}f_{\mu}(x)$
and prove that $(x_{\mu},s_{\mu})\in\PriR^{\circ}\times\DualR^{\circ}$
with $x_{\mu}s_{\mu}=\mu$ for some $s_{\mu}$. Since $\PriR$ is
non-empty and bounded and since $f_{\mu}$ is strictly convex, such
an $x_{\mu}$ exists. Furthermore, since $f_{\mu}(x)\rightarrow+\infty$
as $x_{i}\rightarrow0$ for any $i$, we have that $x_{\mu}\in\PriR^{\circ}$.

By the KKT optimality condition for $f_{\mu}$, there is a vector
$y$ such that
\[
\nabla f_{\mu}(x)=c-\frac{\mu}{x}=\ma^{\top}y.
\]
Define $s_{\mu}=\frac{\mu}{x_{\mu}}$, then one can check that $s_{\mu}\in\DualR^{\circ}$
and $x_{\mu}s_{\mu}=\mu$.

For the uniqueness, if $(x,s)\in\PriR^{\circ}\times\DualR^{\circ}$
and $xs=\mu$, then $x$ satisfies the optimality condition for $f_{\mu}$.
Since $f_{\mu}$ is strictly convex, such $x$ must be unique.
\end{proof}
Lemma \ref{lem:primal_dual_representation} shows that any point in
$\PriR^{\circ}\times\DualR^{\circ}$ is uniquely represented by a
positive vector $\mu$. Interior point methods move $\mu$ uniformly
to $0$ while maintaining the corresponding $x_{\mu}$. Now we discuss
how to find $(x_{\mu},s_{\mu})$ given a nearby interior feasible
point $(x,s)$. Namely, how to move $(x,s)$ to $(x+\delta_{x},s+\delta_{s})$
such that it satisfies the equation
\begin{align*}
(x+\delta_{x})(s+\delta_{s}) & =\mu,\\
\ma(x+\delta_{x}) & =b,\\
\ma^{\top}(y+\delta_{y})+(s+\delta_{s}) & =c,\\
(x+\delta_{x},s+\delta_{s}) & \in\R_{>0}^{2n}.
\end{align*}
Although the equation above involves $y$, our approximate solution
does not need to know $y$. By ignoring the inequality constraint
and the second-order term $\delta_{x}\delta_{s}$ in the first equation
above, and using $\ma x=b$ and $\ma^{\top}y+s=c$ we can simplify
the system:
\begin{align}
xs+\ms\delta_{x}+\mx\delta_{s} & =\mu,\label{eq:LP_newton_step}\\
\ma\delta_{x} & =0,\nonumber \\
\ma^{\top}\delta_{y}+\delta_{s} & =0,\nonumber 
\end{align}
where $\mx$ and $\ms$ are the diagonal matrix with diagonal $x$
and $s$. In the following Lemma, we show how to write the step above
using a projection matrix ($\mproj^{2}=\mproj$).
\begin{lem}
\label{lem:Newton_step}Suppose that $\ma$ has full row rank and
$(x,s)\in\PriR^{\circ}\times\DualR^{\circ}$. Then, the unique solution
for the linear system (\ref{eq:LP_newton_step}) is given by
\begin{align*}
\mx^{-1}\delta_{x} & =(\mi-\mproj)(\delta_{\mu}/\mu),\\
\ms^{-1}\delta_{s} & =\mproj(\delta_{\mu}/\mu)
\end{align*}
where $\delta_{\mu}=\mu-xs$ and $\mproj=\ms^{-1}\ma^{\top}(\ma\ms^{-1}\mx\ma^{\top})^{-1}\ma\mx$.
\end{lem}

\begin{proof}
Note that the step satisfies $\ms\delta_{x}+\mx\delta_{s}=\delta_{\mu}$.
Multiply both sides by $\ma\ms^{-1}$ and using $\ma\delta_{x}=0$,
we have
\[
\ma\ms^{-1}\mx\delta_{s}=\ma\ms^{-1}\delta_{\mu}.
\]
Now we use that $\ma^{\top}\delta_{y}+\delta_{s}=0$ and get
\[
\ma\ms^{-1}\mx\ma^{\top}\delta_{y}=-\ma\ms^{-1}\delta_{\mu}.
\]
Since $\ma\in\R^{d\times n}$ has full row rank and $\ms^{-1}\mx$
is invertible, we have that $\ma\ms^{-1}\mx\ma^{\top}$ is invertible.
Hence, 
\[
\delta_{y}=-(\ma\ms^{-1}\mx\ma^{\top})^{-1}\ma\ms^{-1}\delta_{\mu}\mbox{\ensuremath{\qquad} and \ensuremath{\qquad}}\delta_{s}=\ma^{\top}(\ma\ms^{-1}\mx\ma^{\top})^{-1}\ma\ms^{-1}\delta_{\mu}.
\]
Putting this into $\ms\delta_{x}+\mx\delta_{s}=\delta_{\mu}$, we
have
\[
\delta_{x}=\ms^{-1}\delta_{\mu}-\ms^{-1}\mx\ma^{\top}(\ma\ms^{-1}\mx\ma^{\top})^{-1}\ma\ms^{-1}\delta_{\mu}.
\]
The result follows from the definition of $\mproj$.
\end{proof}

\subsection{Lower Bounding Step Size}

The efficiency of interior point methods depends on how large a step
we can take while staying within the domain. We first study the step
operators $(\mi-\mproj)$ and $\mproj$. The following lemma implies
that $\mproj$ is a nearly orthogonal projection matrix when $\mu$
is close to a multiple of the all-ones vector. Hence, the relative
changes of $\mx^{-1}\delta_{x}$ and $\ms^{-1}\delta_{s}$ are essentially
the orthogonal decomposition of the relative step $\delta_{\mu}/\mu$.
For two vectors $u,v$, we define $\norm u_{v}\defeq\sqrt{u^{\top}\Diag(v)u}$
to be the norm defined by $v$.
\begin{lem}
\label{lem:Projection_property}Under the assumption in Lemma \ref{lem:Newton_step},
$\mproj$ is a projection matrix such that $\|\mproj v\|_{\mu}\leq\|v\|_{\mu}$
for any $v\in\Rn$. Similarly, we have that $\|(\mi-\mproj)v\|_{\mu}\leq\|v\|_{\mu}$.
\end{lem}

\begin{proof}
$\mproj$ is a projection because $\mproj^{2}=\mproj$. Define the
orthogonal projection matrix
\[
\mproj_{\mathrm{orth}}=\ms^{-1/2}\mx^{1/2}\ma^{\top}(\ma\ms^{-1}\mx\ma^{\top})^{-1}\ma\mx^{1/2}\ms^{-1/2},
\]
then we have
\begin{align*}
\|\mproj v\|_{\mu}^{2} & =v^{\top}\mx\ma^{\top}(\ma\ms^{-1}\mx\ma^{\top})^{-1}\ma\ms^{-1}\mx\ms\ms^{-1}\ma^{\top}(\ma\ms^{-1}\mx\ma^{\top})^{-1}\ma\mx v\\
 & =v^{\top}\ms^{1/2}\mx^{1/2}\mproj_{\mathrm{orth}}\ms^{1/2}\mx^{1/2}v\\
 & \leq v^{\top}\ms^{1/2}\mx^{1/2}\ms^{1/2}\mx^{1/2}v=\|v\|_{\mu}^{2}.
\end{align*}
The calculation for $\|(\mi-\mproj)v\|_{\mu}$ is similar.
\end{proof}
Next we give a lower bound on the largest feasible step size.
\begin{lem}
\label{lem:step_size}We have that $\|\mx^{-1}\delta_{x}\|_{\infty}^{2}\leq\frac{1}{\min_{i}\mu_{i}}\|\delta_{\mu}/\mu\|_{\mu}^{2}$
and $\|\ms^{-1}\delta_{s}\|_{\infty}^{2}\leq\frac{1}{\min_{i}\mu_{i}}\|\delta_{\mu}/\mu\|_{\mu}^{2}$.
In particular, if $\|\delta_{\mu}/\mu\|_{\mu}^{2}<\min_{i}\mu_{i}$,
we have $(x+\delta_{x},s+\delta_{s})\in\PriR^{\circ}\times\DualR^{\circ}$
\end{lem}

\begin{proof}
For $\|\mx^{-1}\delta_{x}\|_{\infty}$, we have $\min_{i}\mu_{i}\|\mx^{-1}\delta_{x}\|_{\infty}^{2}\leq\|\mx^{-1}\delta_{x}\|_{\mu}^{2}$
and hence
\[
\|\mx^{-1}\delta_{x}\|_{\infty}^{2}\leq\frac{1}{\min_{i}\mu_{i}}\|\mx^{-1}\delta_{x}\|_{\mu}^{2}=\frac{1}{\min_{i}\mu_{i}}\|(\mi-\mproj)(\delta_{\mu}/\mu)\|_{\mu}^{2}\leq\frac{1}{\min_{i}\mu_{i}}\|\delta_{\mu}/\mu\|_{\mu}^{2}.
\]
The proof for $\|\ms^{-1}\delta_{s}\|_{\infty}$ is similar.

Hence, if $\|\delta_{\mu}/\mu\|_{\mu}^{2}<\min_{i}\mu_{i}$, we have
that $\|\mx^{-1}\delta_{x}\|_{\infty}<1$ and $\|\ms^{-1}\delta_{s}\|_{\infty}<1$,
i.e., $|\delta_{x,i}|<|x_{i}|$ and $|\delta_{s,i}|<|s_{i}|$ for
all $i$. Therefore, $x+\delta_{x}>0$ and $s+\delta_{s}>0$ are feasible.
\end{proof}
To decrease $\mu$ uniformly, we set $\delta_{\mu}=-h\mu$ for some
step size $h$. To ensure the feasibility, we need $\|\delta_{\mu}/\mu\|_{\mu}^{2}\leq\min_{i}\mu_{i}$
and this gives the maximum step size
\begin{equation}
h=\sqrt{\frac{\min_{i}\mu_{i}}{\sum_{i}\mu_{i}}}.\label{eq:max_h_formula}
\end{equation}
Note that the above quantity is maximized at $h=n^{-1/2}$ when $\mu$
has all equal coordinates.

\subsection{Staying within small $\ell_{2}$ distance}

Since the step size (\ref{eq:max_h_formula}) maximizes when $\mu$
is a constant vector. A natural approach is to keep $\mu$ as a vector
close in $\ell_{2}$ norm to a multiple of the all-ones vector. This
motivates the following algorithm:

\begin{algorithm2e}[H]

\caption{$\mathtt{L2Step}(\ma,x,s,t_{\mathrm{start}},t_{\mathrm{end}})$}

\label{alg:L2Step}

\textbf{Define} $\mproj_{x,s}=\ms^{-1}\ma^{\top}(\ma\ms^{-1}\mx\ma^{\top})^{-1}\ma\mx$.

\textbf{Invariant:} $(x,s)\in\PriR^{\circ}\times\DualR^{\circ}$ and
$\|xs-t\|_{2}\leq\frac{t}{4}$. 

Let $t=t_{\mathrm{start}}$, $h=1/(16\sqrt{n})$ and $n$ is the number
of columns in $\ma$.

\Repeat{$t\neq t_{\mathrm{end}}$}{

Let $t'=\max(t/(1+h),t_{\mathrm{end}})$.

Let $\mu=xs$ and $\delta_{\mu}=t'-\mu$.

Let $\delta_{x}=\mx(\mi-\mproj_{x,s})(\delta_{\mu}/\mu)$ and $\delta_{s}=\ms\mproj_{x,s}(\delta_{\mu}/\mu)$.

Set $x\leftarrow x+\delta_{x}$, $s\leftarrow s+\delta_{s}$ and $t\leftarrow t'$.

}

\textbf{Return} $(x,s)$.

\end{algorithm2e}

Note that the algorithm requires some initial point $(x,s)$ close
to the central path and we will show how to get this in the next section
(by changing the linear program temporarily).

First, we show that the invariant is maintained in each step. The
conclusion distance less than $t/6$ is needed in Section \ref{sec:IPM_interior}
where we call $\mathtt{L2Step}$ on a modified LP, then prove the
result is close to central path for the original LP.
\begin{lem}
\label{lem:L2Step}Suppose that the input satisfies $(x,s)\in\PriR^{\circ}\times\DualR^{\circ}$
and $\|xs-t_{\mathrm{start}}\|_{2}\leq\frac{t_{\mathrm{start}}}{4}$.
Then, the algorithm $\mathtt{L2Step}$ maintains $(x,s)\in\PriR^{\circ}\times\DualR^{\circ}$
and $t>0$ such that $\|xs-t\|_{2}\leq\frac{t}{6}$.
\end{lem}

\begin{proof}
We prove by induction that $\|xs-t\|_{2}\leq\frac{t}{6}$ after each
step. Note that the input satisfies $\|xs-t\|_{2}\leq\frac{t}{4}$.

Let $x'=x+\delta_{x}$, $s'=s+\delta_{s}$ and $t'$ defined in the
algorithm. Note that
\[
x's'-t'=(x+\delta_{x})(s+\delta_{s})-t'=\mu+\ms\delta_{x}+\mx\delta_{s}+\delta_{x}\delta_{s}-t'.
\]
Lemma \ref{lem:Newton_step} shows that $\ms\delta_{x}+\mx\delta_{s}=t'-\mu$.
Hence, we have 
\[
x's'-t'=\delta_{x}\delta_{s}=\mx^{-1}\delta_{x}\cdot\ms^{-1}\delta_{s}\cdot\mu.
\]
Using this, we have
\[
\|x's'-t'\|_{2}\leq\|\mu^{1/2}\mx^{-1}\delta_{x}\|_{2}\|\mu^{1/2}\ms^{-1}\delta_{s}\|_{2}=\|\mx^{-1}\delta_{x}\|_{\mu}\|\ms^{-1}\delta_{s}\|_{\mu}\leq\|\delta_{\mu}/\mu\|_{\mu}^{2}.
\]
where we used Lemma \ref{lem:Projection_property} at the end.

Using $t'-\mu=\frac{t'}{t}(t-\mu)+(\frac{t'}{t}-1)\mu$, we have
\[
\|\delta_{\mu}/\mu\|_{\mu}=\|\frac{t'}{t}\frac{t-\mu}{\mu}+(\frac{t'}{t}-1)\|_{\mu}\leq\frac{t'}{t}\|xs-t\|_{\mu^{-1}}+\|\frac{t'}{t}-1\|_{\mu}.
\]
Since $\|\mu-t\|_{2}\leq\frac{t}{4}$, we have $\min_{i}\mu_{i}\geq\frac{3t}{4}$
and $\max_{i}\mu_{i}\leq\frac{5}{4}t$. Using $|\frac{t'}{t}-1|\leq h=\frac{1}{16\sqrt{n}}$,
we have
\[
\|\delta_{\mu}/\mu\|_{\mu}\leq\frac{t'}{t}\sqrt{\frac{4}{3t}}\|xs-t\|_{2}+h\sqrt{\frac{5}{4}tn}\leq\sqrt{\frac{t}{12}}+h\sqrt{\frac{5}{4}tn}\leq0.38\sqrt{t}.
\]
Hence, we have $\|x's'-t'\|_{2}\leq\|\delta_{\mu}/\mu\|_{\mu}^{2}\leq0.15t\leq t'/6$.
Furthermore, $\|\delta_{\mu}/\mu\|_{\mu}^{2}<\min_{i}\mu_{i}$ which
implies $(x,s)$ is feasible (Lemma \ref{lem:step_size}).
\end{proof}
Note that the lemma above only concludes the output is close to central
path. To upper bound the error, we can apply Lemma \ref{lem:duality_gap}
which shows the duality gap is equal to $x^{\top}s$.

\subsection{Solving LP Approximately and Exactly\label{subsec:SolveLP_apx}}

Here we discuss how to get a feasible interior point close to the
central path by modifying the linear program. The runtime of interior
point method depends on how degenerate the linear program is.
\begin{defn}
\label{def:LP_para}We define the following parameters for the linear
program $\min_{\ma x=b,x\geq0}c^{\top}x$:
\begin{enumerate}
\item Inner radius $r$: There exists a $x\in\PriR$ such that $x_{i}\geq r$
for all $i\in[n]$.
\item Outer radius $R$: For any $x\geq0$ with $\ma x=b$, we have that
$\|x\|_{2}\leq R$.
\item Lipschitz constant $L$: $\|c\|_{2}\leq L$.
\end{enumerate}
\end{defn}

Since $\mathtt{L2Step}$ requires a feasible point near the central
path, we modify the linear program to make it happen. To satisfy the
constraint $\ma x=b$, we start the algorithm by taking a least square
solution of the constraint $\ma x=b$. Since it can be negative, we
write the variable $x=x^{+}-x^{-}$ with both $x^{+},x^{-}\geq0$.
We put a large cost vector on $x^{-}$ to ensure the solution is roughly
the same. The crux of the proof is that if we optimize this new program
well enough, we will have $x^{+}-x^{-}>0$ and hence $x^{+}-x^{-}$
gives a good starting point of the original program. Due to technical
reasons, we need to put an extra constraint $1^{\top}x^{+}\leq\Lambda$
for some $\Lambda$ to ensure the problem is bounded. The precise
formulation of the modified linear program is as follows:
\begin{defn}[Modified Linear Program]
\label{def:IPM_interior_modified}Consider a linear program $\min_{\ma x=b,x\geq0}c^{\top}x$
with inner radius $r$, outer radius $R$ and Lipschitz constant $L$.
For any $\overline{R}\geq10R$, $t\geq8L\overline{R}$, we define
the modified primal linear program by 
\[
\min_{(x^{+},x^{-},x^{\theta})\in\mathcal{P}_{\overline{R},t}}c^{\top}x^{+}+\widetilde{c}^{\top}x^{-}
\]
where 
\[
\mathcal{P}_{\overline{R},t}=\{(x^{+},x^{-},x^{\theta})\in\R_{\geq0}^{2n+1}:\ma(x^{+}-x^{-})=b,\sum_{i=1}^{n}x_{i}^{+}+x^{\theta}=\widetilde{b}\}
\]
with $x_{c}^{+}=\frac{t}{c+t/\overline{R}},x_{c}^{-}=x_{c}^{+}-\ma^{\top}(\ma\ma^{\top})^{-1}b,\widetilde{c}=t/x_{c}^{-}$,
$\widetilde{b}=\sum_{i}x_{c,i}^{+}+\overline{R}$. We define the corresponding
dual polytope by
\[
\mathcal{D}_{\overline{R},t}=\{(s^{+},s^{-},s^{\theta})\in\R_{\geq0}^{2n+1}:\ma^{\top}y+\lambda\mathbf{1}+s^{+}=c,-\ma^{\top}y+s^{-}=\widetilde{c},\lambda+s^{\theta}=0\text{ for some }y\in\R^{d}\text{ and }\lambda\in\R\}.
\]
\end{defn}

The main result about the modified program is the following.
\begin{thm}
\label{thm:IPM_interior}Given a linear program $\min_{\ma x=b,x\in\R_{\geq}^{n}}c^{\top}x$
with inner radius $r$, outer radius $R$ and Lipschitz constant $L$.
For any $0\leq\epsilon\leq\frac{1}{2}$, the modified linear program
(Definition \ref{def:IPM_interior_modified}) with $\overline{R}=\frac{5}{\epsilon}R,t=2^{16}\epsilon^{-3}n^{2}\frac{R}{r}\cdot LR$
has the following properties:
\begin{itemize}
\item The point $(x_{c}^{+},x_{c}^{-},\overline{R})$ is on the central
path of the modified program at $t$.
\item For any primal $x\defeq(x^{+},x^{-},x^{\theta})\in\mathcal{P}_{\overline{R},t}$
and dual $s\defeq(s^{+},s^{-},s^{\theta})\in\mathcal{D}_{\overline{R},t}$
such that $\frac{5}{6}LR\leq x_{i}s_{i}\leq\frac{7}{6}LR$, we have
that
\[
(x^{+}-x^{-},s^{+}-s^{\theta})\in\PriR\times\DualR
\]
and that $x_{i}^{-}\leq\epsilon x_{i}^{+}$ and $s^{\theta}\leq\epsilon s_{i}^{+}$
for all $i$.
\end{itemize}
\end{thm}

\begin{proof}
Since the proof is not illuminating, we defer it to Appendix \ref{sec:IPM_interior}
(Lemma \ref{lem:ipm_interior_modified} and Lemma \ref{lemma:IPM_center_original}).
\end{proof}
Now we state our main algorithm.

\begin{algorithm2e}[H]

\caption{$\mathtt{SlowSolveLP}(\ma,b,c,x^{(0)},\delta)$}

\label{alg:SlowSolveLP}

\textbf{Assumption: the linear program has inner radius $r$, outer
radius $R$ and Lipschitz constant $L$.}

Let $\epsilon=1/(100\sqrt{n})$, $\overline{R}=\frac{5}{\epsilon}R$,
$t=2^{16}\epsilon^{-3}n^{2}\frac{R}{r}\cdot LR$. 

\tcp{Define the modified program $\min_{\oma x=\overline{b}}\overline{c}^{\top}x$
by Definition \ref{def:IPM_interior_modified} with parameters $\overline{R}$
and $t$.}

Let $\oma=\left[\begin{array}{ccc}
\ma & -\ma & 0\\
1 & 0 & 1
\end{array}\right]$, $\overline{c}=(c,\widetilde{c})$, $\overline{b}=(b,\widetilde{b})$
where $\widetilde{c}$ and $\widetilde{b}$ are defined in Definition
\ref{def:IPM_interior_modified}.

\tcp{Write down the central path at $t$ for modified linear program
using Lemma \ref{lem:ipm_interior_modified}.}

$\overline{x}=(x_{c}^{+},x_{c}^{-},\overline{R})$. $\overline{s}=x/t$.

$(\overline{x},\overline{s})=\mathtt{L2Step}(\oma,\overline{x},\overline{s},t,LR)$.

$(x,s)=(x^{+}-x^{-},s^{+}-s^{\theta})$ where $\overline{x}=(x^{+},x^{-},x^{\theta})$
and $\overline{s}=(s^{+},s^{-},s^{\theta})$.

$(x_{\mathrm{end}},s_{\mathrm{end}})=\mathtt{L2Step}(\ma,x^{+}-x^{-},s^{+}-s^{\theta},LR,t_{\mathrm{end}})$
with $t_{\mathrm{end}}=\delta LR/(2n)$.

\textbf{Return} $x_{\mathrm{end}}$.

\end{algorithm2e}
\begin{thm}
\label{thm:SlowSolveLP}Consider a linear program $\min_{\ma x=b,x\geq0}c^{\top}x$
with $n$ variables and $d$ constraints. Assume the linear program
has inner radius $r$, outer radius $R$ and Lipschitz constant $L$
(see Definition \ref{def:LP_para}). Then, $\mathtt{SlowSolveLP}$
outputs $x$ such that
\begin{align*}
c^{\top}x & \leq\min_{\ma x=b,x\geq0}c^{\top}x+\delta LR,\\
\ma x & =b,\\
x & \geq0.
\end{align*}
The algorithm takes $O(\sqrt{n}\log(nR/(\delta r)))$ Newton steps
(defined in (\ref{eq:LP_newton_step})).

If we further assume that the solution $x^{*}=\arg\min_{\ma x=b,x\geq0}c^{\top}x$
is unique and that $c^{\top}x\geq c^{\top}x^{*}+\eta LR$ for any
other vertex $x$ of $\{\ma x=b,x\geq0\}$ for some $\eta>\delta\geq0$,
then we have that $\|x-x^{*}\|_{2}\leq\frac{2\delta R}{\eta}$.
\end{thm}

\begin{proof}
By Theorem \ref{thm:IPM_interior}, the point $(x_{c}^{+},x_{c}^{-},\overline{R})$
is on the central path of the modified program at $t$. After the
first call of $\mathtt{L2Step}$, Lemma \ref{lem:L2Step} shows that
$\mathtt{L2Step}$ returns $(\overline{x},\overline{s})$ such that
$\|\overline{x}\overline{s}-t\|_{2}\leq\frac{t}{6}$ with $t=LR$.

Theorem \ref{thm:IPM_interior} shows that $(x,s)=(x^{+}-x^{-},s^{+}-s^{\theta})\in\PriR\times\DualR$
and that $x=(1\pm\epsilon)x^{+}$ and $s=(1\pm\epsilon)s^{+}$. Since
$\epsilon=\frac{1}{100\sqrt{n}}$ and $\|x^{+}s^{+}-t\|_{2}\leq\frac{t}{6}$
with $t=LR$, we have that $\|xs-t\|_{2}\leq\frac{t}{4}$. This verifies
the condition for the second call of $\mathtt{L2Step}$.

After the second call of $\mathtt{L2Step}$, Lemma \ref{lem:L2Step}
shows that $\mathtt{L2Step}$ returns $(x_{\mathrm{end}},s_{\mathrm{end}})$
such that $\|x_{\mathrm{end}}s_{\mathrm{end}}-t_{\mathrm{end}}\|_{2}\leq\frac{t}{6}$
with $t_{\mathrm{end}}=\delta LR/(2n)$. Hence, Lemma \ref{lem:duality_gap}
shows that 
\[
c^{\top}x_{\mathrm{end}}\leq\min_{\ma x=b,x\geq0}c^{\top}x+x_{\mathrm{end}}^{\top}s_{\mathrm{end}}\leq\min_{\ma x=b,x\geq0}c^{\top}x+2t_{\mathrm{end}}n\leq\min_{\ma x=b,x\geq0}c^{\top}x+\delta LR.
\]
For the runtime, note that $\mathtt{L2Step}$ decreases $t$ by $1-\Omega(n^{-1/2})$
factor each step. Hence, the first call takes $O(\sqrt{n}\log(nR/r))$
Newton steps and second call takes $O(\sqrt{n}\log(n/\delta))$ Newton
steps.

For the last conclusion, we assume $\delta\leq\eta$ and let $\mathcal{P}_{t}=\mathcal{P}\cap\{c^{\top}x\leq c^{\top}x^{*}+tLR\}$.
Note that $\mathcal{P}_{\eta}$ is a cone at $x^{*}$ (because there
is no vertex except $x^{*}$ with value less than $c^{\top}x^{*}+tLR$).
Hence, we have $\mathcal{P}_{\delta}-x^{*}=\frac{\delta}{\eta}(\mathcal{P}_{\eta}-x^{*})$.
Since $x\in\mathcal{P}_{\delta}$, we have that 
\[
\|x-x^{*}\|_{2}\leq\frac{\delta}{\eta}\text{diameter}(\mathcal{P}_{\eta}-x^{*})\leq\frac{2\delta R}{\eta}.
\]
\end{proof}
If we know the solution of the linear program is integral or rational
with some bound on the number of bits, then getting a solution close
enough to $x^{*}$ allows us to round the solution to an integral
solution. Therefore, the last conclusion of the theorem above gives
us an exact linear program algorithm assuming $\ma,b,c$ are integral
and bounded. The uniqueness assumption can be achieved by perturbing
the cost vector by a random vector (e.g., using the ``isolation''
lemma \cite[Lemma 4]{klivans2001randomness}).
\begin{xca}
Make the perturbation deterministic while preserving solutions.
\end{xca}

\section{Robust Interior Point Method}

To improve the interior point method, one can either improve the number
of steps $\widetilde{O}(\sqrt{n})$ or the cost per step. The first
is a major open problem. In this note, we focus on the latter question.
Recall from (\ref{eq:LP_newton_step}) that the linear system we solve
in each step is of the form
\begin{align}
\ms\delta_{x}+\mx\delta_{s} & =\delta_{\mu},\nonumber \\
\ma\delta_{x} & =0,\nonumber \\
\ma^{\top}\delta_{y}+\delta_{s} & =0.\label{eq:LP_newton_step_2}
\end{align}
In each step, $x,s$ and $\delta_{\mu}$ in the equation above changes
relatively by a vector with bounded $\ell_{2}$ norm. So, only few
coordinates change a lot in each step. To take advantage of this,
the robust interior point method contains two new components: 1) Analyze
the convergence when we only solve the linear system approximately
(Section \ref{subsec:IPM_cosh_step}). 2) Show how to maintain the
solution throughout the iteration (Section \ref{subsec:select_xsr}
and Section \ref{subsec:IPM_Inverse_Maintenance}).

\subsection{Staying within small $\ell_{\infty}$ distance\label{subsec:IPM_cosh_step}}

In the above description and analysis, we assumed that we computed
each step of the interior point method precisely. But one can imagine
that it suffices to compute steps approximately since our goal is
only to stay close to the central path. This could have significant
computational advantages.

To make the interior point method robust to noise in the updates to
$x$ and $s$, we need the method to work under a larger neighborhood
than that given by the Euclidean norm ($\|xs-t\|_{2}\leq\frac{t}{4}$).
We cannot increase the radius of the $\ell_{2}$ ball because we need
the neighborhood to lie strictly inside the feasible region. One natural
alternative choice of distance and potential would be a higher norm,
$\|xs-t\|_{q}^{q}$. However, analyzing the step $\delta_{\mu}$ that
minimizes $\|\mu+\delta_{\mu}-t\|_{q}^{q}$ involves many cases. Instead,
we use the potential
\begin{equation}
\Phi(r)=\sum_{i=1}^{n}\cosh(\lambda r_{i})=\sum_{i=1}^{n}\frac{(e^{\lambda r_{i}}+e^{-\lambda r_{i}})}{2}.\label{eq:Phi_potential}
\end{equation}
with $r=\frac{xs-t}{t}$ for some scalar $\lambda=\Theta(\log n)$.
This potential induces the following algorithm where each step of
the algorithm takes the step $\delta_{\mu}\approx-c\nabla\Phi(\frac{xs-t}{t})$.

\begin{algorithm2e}[H]

\caption{$\mathtt{RobustStep}(\ma,x,s,t_{\mathrm{start}},t_{\mathrm{end}})$}

\label{alg:RobustStep}

\textbf{Define} $\Phi(r)$ and $r$ according to (\ref{eq:Phi_potential})
with $\lambda=16\log40n$.

\textbf{Invariant:} $(x,s)\in\PriR^{\circ}\times\DualR^{\circ}$ and
$\Phi(r)\leq16n$.

Let $t=t_{\mathrm{start}}$, $h=1/(128\lambda\sqrt{n})$ and $n$
is the number of columns in $\ma$.

\Repeat{$t\neq t_{\mathrm{end}}$}{

Pick $\ox$, $\os$ and $\overline{r}$ such that $\|\ln\ox-\ln x\|_{\infty}\leq\frac{1}{48}$,
$\|\ln\os-\ln s\|_{\infty}\leq\frac{1}{48}$ and $\|\overline{r}-r\|_{\infty}\leq\frac{1}{48\lambda}$.

Let $t'=\max(t/(1+h),t_{\mathrm{end}})$, $\overline{\delta}_{\mu}=-\frac{t'}{32\lambda}\frac{\overline{g}}{\|\overline{g}\|_{2}}$,
$\overline{g}=\nabla\Phi(\overline{r})$.

Find $\delta_{x},\delta_{s}$ such that

\begin{align}
\oms\delta_{x}+\omx\delta_{s} & =\overline{\delta}_{\mu},\nonumber \\
\ma\delta_{x} & =0,\nonumber \\
\ma^{\top}\delta_{y}+\delta_{s} & =0.\label{eq:LP_newton_step_2-1}
\end{align}

Set $x\leftarrow x+\delta_{x}$, $s\leftarrow s+\delta_{s}$ and $t\leftarrow t'$.

}

\textbf{Return} $(x,s)$.

\end{algorithm2e}

\begin{figure}
\begin{centering}
\includegraphics[height=2in]{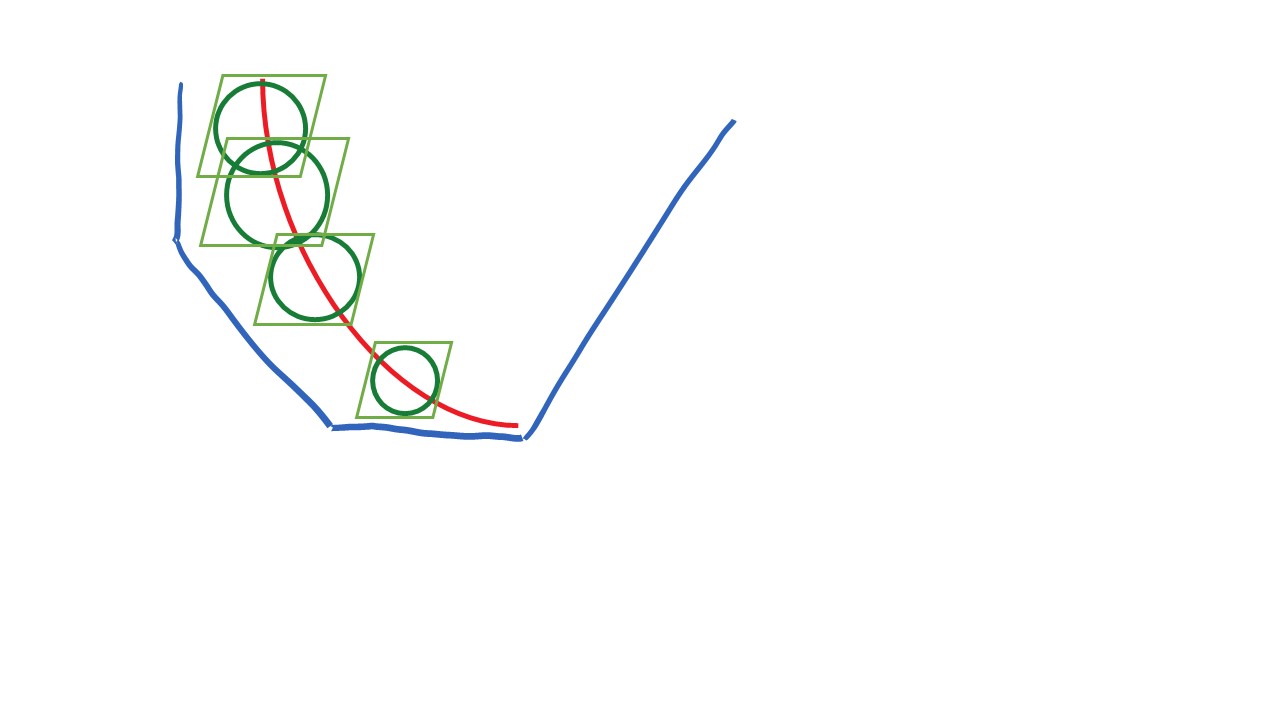}
\par\end{centering}
\caption{The Central Path. The standard method keeps the current point within
a relative $\ell_{2}$-norm ball, while it suffices to use a larger
$\ell_{\infty}$-norm ball. The robust method provides a bridge to
the latter via the cosh ball defined by $\Phi$. }
\end{figure}

We begin with useful facts about $\Phi$.
\begin{lem}
\label{lem:robust_IPM_Phi}Define $\Phi(r)$ according to (\ref{eq:Phi_potential}).
For any $r\in\Rn$, we have that $\|r\|_{\infty}\leq\frac{\log2\Phi(r)}{\lambda}$
and $\|\nabla\Phi(r)\|_{2}\geq\frac{\lambda}{\sqrt{n}}(\Phi(r)-n)$.
Moreover, if $\Phi(r)\geq4n$ and $\|\delta\|_{\infty}\leq\frac{1}{5\lambda}$,
we have 
\[
\|\nabla\Phi(r+\delta)-\nabla\Phi(r)\|_{2}\leq\frac{1}{3}\|\nabla\Phi(r)\|_{2}.
\]
\end{lem}

\begin{proof}
We have $\Phi(r)\geq\frac{1}{2}\min_{i}e^{\lambda|r_{i}|}$ and hence
$\|r\|_{\infty}\leq\frac{\log2\Phi(r)}{\lambda}$.

For the second claim, wsing that $\nabla\Phi(r)=\sum_{i=1}^{n}\lambda\sinh(\lambda r_{i})$,
we have
\begin{align*}
\|\nabla\Phi(r)\|_{2} & =\lambda\sqrt{\sum_{i=1}^{n}\sinh^{2}(\lambda r_{i})}=\lambda\sqrt{\sum_{i=1}^{n}(\cosh^{2}(\lambda r_{i})-1)}\\
 & \geq\frac{\lambda}{\sqrt{n}}\sum_{i=1}^{n}\sqrt{\cosh^{2}(\lambda r_{i})-1}\geq\frac{\lambda}{\sqrt{n}}\sum_{i=1}^{n}(\cosh(\lambda r_{i})-1)\\
 & =\frac{\lambda}{\sqrt{n}}(\Phi(r)-n).
\end{align*}

For the last claim, using $\sinh(r+\delta)=\sinh r\cosh\delta+\cosh r\sinh\delta$
and $|\cosh r-\sinh r|\leq1$, for $|\delta|\leq\frac{1}{5}$, we
have
\begin{align*}
|\sinh(r+\delta)-\sinh(r)| & \leq|\cosh\delta-1|\cdot|\sinh r|+|\sinh\delta|\cdot\cosh r\\
 & \leq\left(|\cosh\delta-1|+|\sinh\delta|\right)\cdot|\sinh r|+|\sinh\delta|\\
 & \leq\frac{1}{4}|\sinh r|+\frac{1}{4}.
\end{align*}
Using that $\nabla\Phi(r)=\sum_{i=1}^{n}\lambda\sinh(\lambda r_{i})$,
for $\|\delta\|_{\infty}\leq\frac{1}{5\lambda}$, we have
\begin{equation}
\|\nabla\Phi(r+\delta)-\nabla\Phi(r)\|_{2}\leq\frac{1}{4}\|\nabla\Phi(r)\|_{2}+\frac{\sqrt{n}\lambda}{4}.\label{eq:grad_phi_diff}
\end{equation}
Since $\Phi(r)\geq4n$, we have that $\|\nabla\Phi(r)\|_{2}\geq3\sqrt{n}\lambda$
and hence (\ref{eq:grad_phi_diff}) shows that
\[
\|\nabla\Phi(r+\delta)-\nabla\Phi(r)\|_{2}\leq(\frac{1}{4}+\frac{1}{12})\|\nabla\Phi(r)\|_{2}=\frac{1}{3}\|\nabla\Phi(r)\|_{2}.
\]
\end{proof}
We collect some basic bounds on the step in the following lemma.
\begin{lem}
\label{lem:RobustIPM_basic}Using the notation in $\mathtt{RobustStep}$
(Algorithm \ref{alg:RobustStep}). Under the invariant $\Phi((xs-t)/t)\leq16n$,
we have $\|xs-t\|_{\infty}\leq\frac{t}{16}$, $\|\delta_{x}/x\|_{2}\leq\frac{1}{16\lambda}$,
and $\|\delta_{s}/s\|_{2}\leq\frac{1}{16\lambda}$.
\end{lem}

\begin{proof}
Using $\Phi((xs-t)/t)\leq16n$ and Lemma \ref{lem:robust_IPM_Phi},
we have
\[
\|xs-t\|_{\infty}\leq\frac{t\log32n}{\lambda}\leq\frac{t}{16}
\]

By Lemma \ref{lem:Newton_step}, we have $\mx^{-1}\delta_{x}=(\mi-\mproj)(\overline{\delta}_{\mu}/\overline{\mu})$
where $\overline{\mu}=\ox\os$ and $\mproj=\oms^{-1}\ma^{\top}(\ma\oms^{-1}\omx\ma^{\top})^{-1}\ma\omx$.
By Lemma \ref{lem:Projection_property}, we have 
\[
\|\delta_{x}/x\|_{\overline{\mu}}=\|(\mi-\mproj)v\|_{\overline{\mu}}\leq\|\overline{\delta}_{\mu}/\overline{\mu}\|_{\overline{\mu}}.
\]
Using that $\|xs-t\|_{\infty}\leq\frac{t}{16}$, $\|\ln\ox-\ln x\|_{\infty}\leq\frac{1}{48}$,
$\|\ln\os-\ln s\|_{\infty}\leq\frac{1}{48}$, we have $\overline{\mu}\geq\frac{10}{11}t$
and hence
\[
\|\delta_{x}/x\|_{2}\leq\sqrt{\frac{11}{10t}}\|\delta_{x}/x\|_{\overline{\mu}}\leq\sqrt{\frac{11}{10t}}\|\overline{\delta}_{\mu}\|_{\overline{\mu}^{-1}}\leq\frac{11}{10t}\|\overline{\delta}_{\mu}\|_{2}
\]
Using the formula $\overline{\delta}_{\mu}=-\frac{t'}{32\lambda}\frac{\overline{g}}{\|\overline{g}\|_{2}}$,
we have
\begin{align*}
\|\delta_{x}/x\|_{2} & \leq\frac{11}{10}\frac{t'}{32\lambda t}\leq\frac{1}{16\lambda}.
\end{align*}
Same proof gives $\|\delta_{s}/s\|_{2}\leq\frac{1}{16\lambda}$.
\end{proof}
Using this, we prove the algorithm $\mathtt{RobustStep}$ satisfies
the invariant on the distance.
\begin{lem}
\label{lem:RobustStep}Suppose that the input satisfies $(x,s)\in\PriR^{\circ}\times\DualR^{\circ}$
and $\Phi((xs-t_{\mathrm{start}})/t_{\mathrm{start}})\leq16n$. Let
$x^{(k)},s^{(k)},t^{(k)}$ be the $x,s,t$ computed in the $\mathtt{RobustStep}$
after the $k$-th step. Let $\Phi^{(k)}=\Phi((x^{(k)}s^{(k)}-t^{(k)})/t^{(k)})$.
Then, we have 
\[
\Phi^{(k+1)}\leq\begin{cases}
12n & \text{if }\Phi^{(k)}\leq8n\\
\Phi^{(k)} & \text{otherwise}
\end{cases}.
\]
Furthermore, we have that $\|r^{(k+1)}-r^{(k)}\|_{2}\leq\frac{1}{16\lambda}$
where $r^{(k)}=(x^{(k)}s^{(k)}-t^{(k)})/t^{(k)}$.
\end{lem}

\begin{proof}
Fix some iteration $k$. Let $x=x^{(k)},s=s^{(k)},t=t^{(k)},x'=x^{(k+1)},s'=s^{(k+1)}$
and $t'=t^{(k+1)}$. We define $r=(xs-t)/t$ and $r'=(x's'-t')/t'$.
By the definition of $\delta_{x}$ and $\delta_{s}$, we have $\oms\delta_{x}+\omx\delta_{s}=\overline{\delta}_{\mu}=-\frac{t'}{32\lambda}\frac{\overline{g}}{\|\overline{g}\|_{2}}$
and hence
\begin{align}
\frac{x's'-t'}{t'}= & \frac{(x+\delta_{x})(s+\delta_{s})-t'}{t'}=\frac{xs+s\delta_{x}+x\delta_{s}+\delta_{x}\delta_{s}-t'}{t'}\nonumber \\
= & \frac{xs-t'+\os\delta_{x}+\ox\delta_{s}+(s-\os)\delta_{x}+(x-\ox)\delta_{s}+\delta_{x}\delta_{s}}{t'}\nonumber \\
= & \frac{xs-t}{t}-\frac{1}{32\lambda}\cdot\frac{\overline{g}}{\|\overline{g}\|_{2}}+\eta\label{eq:IPM_robust_r_diff}
\end{align}
where the error term
\[
\eta=(\frac{t}{t'}-1)\frac{xs}{t}+\frac{(s-\os)\delta_{x}+(x-\ox)\delta_{s}+\delta_{x}\delta_{s}}{t'}.
\]
Now, we bound the error term $\eta$. Using Lemma \ref{lem:RobustIPM_basic}
($\|\delta_{x}/x\|_{2}\leq\frac{1}{16\lambda}$, $\|\delta_{s}/s\|_{2}\leq\frac{1}{16\lambda}$,
$\|xs-t\|_{\infty}\leq\frac{t}{16}$) and the definition of the algorithm
($\lambda\geq16$, $|t'-t|\leq\frac{t'}{128\lambda\sqrt{n}}\leq\frac{t}{128\lambda\sqrt{n}}$,
$\|\ln\ox-\ln x\|_{\infty}\leq\frac{1}{48}$, $\|\ln\os-\ln s\|_{\infty}\leq\frac{1}{48}$),
we have
\begin{align}
\|\eta\|_{2}\leq & |\frac{t}{t'}-1|\|\frac{xs}{t}\|_{\infty}\sqrt{n}+\|\frac{xs}{t'}\|_{\infty}\|\frac{s-\os}{s}\|_{\infty}\|\frac{\delta_{x}}{x}\|_{2}\nonumber \\
 & +\|\frac{xs}{t'}\|_{\infty}\|\frac{x-\ox}{x}\|_{\infty}\|\frac{\delta_{s}}{s}\|_{2}+\|\frac{xs}{t'}\|_{\infty}\|\frac{\delta_{x}}{x}\|_{2}\|\frac{\delta_{s}}{s}\|_{2}\nonumber \\
\leq & \frac{1}{128\lambda}\frac{17}{16}+\frac{9}{8}(e^{1/48}-1)(\frac{1}{16\lambda}+\frac{1}{16\lambda})+\frac{9}{8}(\frac{1}{16\lambda}\frac{1}{16\lambda})\leq\frac{1}{60\lambda}.\label{eq:robust_IPM_eta}
\end{align}

In particular, we use (\ref{eq:IPM_robust_r_diff}) and (\ref{eq:robust_IPM_eta})
to get
\[
\|r-r'\|_{2}\leq\frac{1}{32\lambda}+\|\eta\|_{2}\leq\frac{1}{16\lambda}.
\]
This proves the conclusion about $r$.

Case 1: $\Phi(r)\leq8n$. 

The definition of $\Phi$ together with the fact $\|r-r'\|_{2}\leq\frac{1}{16\lambda}$
implies that $\Phi(r')\leq\frac{3}{2}\Phi(r)\leq12n$. 

Case 2: $\Phi(r)\geq8n$.

Mean value theorem shows there is $\widetilde{r}$ between $r$ and
$r'$ such that
\[
\Phi(r')=\Phi(r)+\left\langle \nabla\Phi(\widetilde{r}),r'-r\right\rangle =\Phi(r)+\left\langle \nabla\Phi(\widetilde{r}),-\frac{1}{32\lambda}\frac{\overline{g}}{\|\overline{g}\|_{2}}+\eta\right\rangle 
\]
where we used (\ref{eq:IPM_robust_r_diff}) at the end. Using $\|r-r'\|_{2}\leq\frac{1}{16\lambda}$
and $\|\overline{r}-r\|_{\infty}\leq\frac{1}{48\lambda}$ (by assumption),
we have $\|\overline{r}-\widetilde{r}\|_{\infty}\leq\frac{1}{5\lambda}$.
Since $\Phi(r)\geq8n$, we have $\Phi(\overline{r})\geq4n$ and hence
Lemma \ref{lem:robust_IPM_Phi} shows that
\[
\|\nabla\Phi(\widetilde{r})-\nabla\Phi(\overline{r})\|_{2}\leq\frac{1}{3}\|\nabla\Phi(\overline{r})\|_{2}.
\]
Using $\overline{g}=\nabla\Phi(\overline{r})$ and letting $\eta_{2}=\nabla\Phi(\widetilde{r})-\nabla\Phi(\overline{r})$,
we have
\[
\Phi(r')-\Phi(r)=\left\langle \overline{g}+\eta_{2},-\frac{1}{32\lambda}\frac{\overline{g}}{\|\overline{g}\|_{2}}+\eta\right\rangle =-\frac{1}{32\lambda}\|\overline{g}\|_{2}-\frac{1}{32\lambda}\eta_{2}^{\top}\frac{\overline{g}}{\|\overline{g}\|_{2}}+\overline{g}^{\top}\eta+\eta_{2}^{\top}\eta.
\]
Using $\|\eta_{2}\|_{2}\leq\frac{1}{3}\|\overline{g}\|_{2}$ and $\|\eta\|_{2}\leq\frac{1}{60\lambda}$
(\ref{eq:robust_IPM_eta}), we have
\begin{align*}
\Phi(r')-\Phi(r) & \leq-\frac{1}{32\lambda}\|\overline{g}\|_{2}+\frac{1}{32\lambda}\cdot\frac{1}{3}\|\overline{g}\|_{2}+\|\overline{g}\|_{2}\cdot\frac{1}{60\lambda}+\frac{1}{3}\|\overline{g}\|_{2}\cdot\frac{1}{60\lambda}\leq-\frac{1}{720\lambda}\|\overline{g}\|_{2}
\end{align*}
Using Lemma \ref{lem:robust_IPM_Phi}, we have $\|\overline{g}\|_{2}\geq\frac{\lambda}{\sqrt{n}}(\Phi(\overline{r})-n)\geq3\lambda\sqrt{n}$.
Hence, we have
\begin{equation}
\Phi(r')\leq\Phi(r)-\frac{\sqrt{n}}{240}<16n.\label{eq:IPM_robust_decrease}
\end{equation}
The potential actually decreases in this case. This proves the conclusion
about $\Phi$.
\end{proof}
\begin{xca}
Show that we can use the potential function $\Phi=\sum_{i=1}^{n}\exp(\lambda|v_{i}|)$
to get a similar conclusion.
\end{xca}

\subsection{Selecting $\protect\ox,\protect\os$ and $\overline{r}$\label{subsec:select_xsr}}

Each step of the robust interior point method solves the linear system
\[
\left(\begin{array}{ccc}
\oms & \omx & \mzero\\
\ma & \mzero & \mzero\\
\mzero & \mi & \ma^{\top}
\end{array}\right)\left(\begin{array}{c}
\delta_{x}\\
\delta_{s}\\
\delta_{y}
\end{array}\right)=\left(\begin{array}{c}
\nabla\Phi(\overline{r})\\
0\\
0
\end{array}\right)
\]
for some vectors $\ox,\os,\overline{r}$ such that $\|\ln\ox-\ln x\|_{\infty}\leq\frac{1}{48}$,
$\|\ln\os-\ln s\|_{\infty}\leq\frac{1}{48}$, $\|\overline{r}-r\|_{\infty}\leq\frac{1}{48\lambda}$.
The key observation is that only a few coordinates of $x,s$ and $r$
change significantly each step and hence we can maintain the solution
of the linear system instead of computing from scratch. In this section,
we discuss how to select $\ox,\os,\overline{r}$ with as few updates
as possible while maintaining the invariants.

First, we observe that $\ln x,\ln s$ and $r$ change by $O(1)$ in
$\ell_{2}$ norm in each step.
\begin{lem}
\label{lem:change_xsr}Define $x^{(k)},s^{(k)},r^{(k)}$ according
to Lemma \ref{lem:RobustStep}. Then, $\|\ln x^{(k+1)}-\ln x^{(k)}\|_{2}$,
$\|\ln s^{(k+1)}-\ln s^{(k)}\|_{2}$ and $\|r^{(k+1)}-r^{(k)}\|_{2}$
are all bounded by $1/(8\lambda)$.
\end{lem}

\begin{proof}
Lemma \ref{lem:Projection_property} shows that
\begin{equation}
\|(x^{(k+1)}-x^{(k)})/x^{(k)}\|_{\mu^{(k)}}\leq\|\overline{\delta}_{\mu}/\mu^{(k)}\|_{\mu^{(k)}}\label{eq:robust_x_step}
\end{equation}
where $\mu^{(k)}=\overline{x}^{(k)}\overline{s}^{(k)}$ and $\overline{x}^{(k)},\overline{s}^{(k)}$
are the $\ox,\os$ used in the $k$-th step. 

To bound $\mu^{(k)}$, Lemma \ref{lem:RobustStep} shows that the
invariant $\Phi(r^{(k)})\le16n$ holds and hence Lemma \ref{lem:robust_IPM_Phi}
shows that (recall $\lambda=16\log40n$):
\[
\|(x^{(k)}s^{(k)}-t^{(k)})/t^{(k)}\|_{\infty}=\|r^{(k)}\|_{\infty}\leq\frac{\log32n}{\lambda}\leq\frac{1}{16}.
\]
 Together with the fact that $\|\ln\ox^{(k)}-\ln x^{(k)}\|_{\infty}\leq\frac{1}{48}$,
$\|\ln\os^{(k)}-\ln s^{(k)}\|_{\infty}\leq\frac{1}{48}$, we have
\[
\|\frac{\ox^{(k)}\os^{(k)}-t^{(k)}}{t^{(k)}}\|_{\infty}\leq\frac{1}{8}.
\]
Using this on (\ref{eq:robust_x_step}) gives $\|(x^{(k+1)}-x^{(k)})/x^{(k)}\|_{2}\leq\frac{8}{7}\frac{1}{t^{(k)}}\|\overline{\delta}_{\mu}\|_{2}.$
Using $\overline{\delta}_{\mu}=-\frac{t'}{32\lambda}\frac{\overline{g}}{\|\overline{g}\|_{2}}$,
we have
\[
\|(x^{(k+1)}-x^{(k)})/x^{(k)}\|_{2}\leq\frac{1}{28\lambda}.
\]
To translate the bound to $\log$ scale, we note that $|\ln(1+t)-t|\leq2t$
for all $|t|\leq\frac{1}{2}$ and hence
\[
\|\ln x^{(k+1)}-\ln x^{(k)}\|_{2}=\|\ln(1+\frac{x^{(k+1)}-x^{(k)}}{x^{(k)}})\|_{2}\leq\frac{1}{14\lambda}.
\]

The bound for $\|\ln s^{(k+1)}-\ln s^{(k)}\|_{2}$ is similar.

The bound for $\|r^{(k+1)}-r^{(k)}\|_{2}$ follows from Lemma \ref{lem:RobustStep}.
\end{proof}
Now the question is how to select $\ln\ox,\ln\os$ and $\overline{r}$
such that they are close to $\ln x,\ln s$ and $r$ in $\ell_{\infty}$
norm. If the cost of updating the inverse of a matrix is linear in
the rank of the update, then we can simply update any coordinate of
$\ox,\os$ and $\overline{r}$ whenever they violate the condition.
However, due to fast matrix multiplication, the average cost (per
rank) of update is lower when the rank of update is large. Therefore,
it is beneficial to update coordinates preemptively. 

Now we state the algorithm for selecting $\ox,\os$ and $\overline{r}$.
This algorithm is a general algorithm for maintaining a vector $\overline{v}$
such that $\|\overline{v}-v\|_{\infty}\leq\delta$. For every $2^{k}$
steps, the algorithm updates the coordinate of the vector $\overline{v}$
if that coordinate has changed by more than $\delta/(2\log n)$ between
this step and $2^{k}$ steps earlier.

\begin{algorithm2e}[H]

\caption{$\mathtt{SelectVector}(\overline{v},v^{(0)},v^{(1)},\cdots,v^{(k)},\delta)$}

\label{alg:SelectVector}

Let $I=\{\}$ be the set of updating coordinates.

\For{$\ell=0,1,\cdots,\left\lceil \log n\right\rceil $}{

\If{$k=0\mod2^{\ell}$}{

\uIf{$\ell=\left\lceil \log n\right\rceil $}{

$I=[n]$.

}\Else{

$I=I\cup\{i:|v_{i}^{(k)}-v_{i}^{(k-2^{\ell})}|\geq\delta/(2\left\lceil \log n\right\rceil )\}$.

}

}

}

$\overline{v}_{i}\leftarrow v_{i}^{(k)}$ for all $i\in I$

\textbf{Return} $\overline{v}$

\end{algorithm2e}
\begin{lem}
\label{lem:SelectVector}Given vectors $v^{(0)},v^{(1)},v^{(2)},\cdots$
arriving in a stream, suppose that $\|v^{(k+1)}-v^{(k)}\|_{2}\leq\beta$
for all $k$. For any $\frac{1}{2}>\delta>0$, define the vector $\overline{v}^{(0)}=v^{(0)}$
and $\overline{v}^{(k)}=\mathtt{SelectVector}(\overline{v}^{(k-1)},v^{(0)},v^{(1)},\cdots,v^{(k)},\delta)$.
Then, we have that
\begin{itemize}
\item $\|\overline{v}^{(k)}-v^{(k)}\|_{\infty}\leq\delta$ for all $k$.
\item $\|\overline{v}^{(k)}-\overline{v}^{(k-1)}\|_{0}\leq O(2^{2\ell_{k}}(\beta/\delta)^{2}\log^{2}n)$
where $\ell_{k}$ is the largest integer $\ell$ with $k=0\mod2^{\ell}$. 
\end{itemize}
\end{lem}

\begin{proof}
For bounding the error, we first fix some coordinate $i\in[n]$. Let
$k'$ be the iteration when $\overline{v}_{i}$ was last updated,
namely, $\overline{v}_{i}^{(k)}=\overline{v}_{i}^{(k')}=v_{i}^{(k')}$.
Since we set $\overline{v}\leftarrow v$ every $2^{\left\lceil \log n\right\rceil }$
steps, we have $k-2^{\left\lceil \log n\right\rceil }\leq k'<k$.
We can write $k'=k_{0}<k_{1}<k_{2}<\cdots<k_{s}=k$ such that $k_{i+1}-k_{i}$
is a power of $2$ and $k_{i+1}-k_{i}$ divides $k_{i+1}$ with $|s|\leq2\left\lceil \log n\right\rceil $.
Hence, we have that
\[
v_{i}^{(k)}-\overline{v}_{i}^{(k)}=v_{i}^{(k_{s})}-v_{i}^{(k_{0})}=\sum_{j=0}^{s-1}(v_{i}^{(k_{j+1})}-v_{i}^{(k_{j})}).
\]
Since $\overline{v}_{i}$ is not updated since step $k'$, we have
$|v_{i}^{(k_{j+1})}-v_{i}^{(k_{j})}|\leq\delta/(2\left\lceil \log n\right\rceil )$
and hence $|v_{i}^{(k)}-\overline{v}_{i}^{(k)}|\leq\delta$. Since
this holds for every $i$, we have that $\|\overline{v}^{(k)}-v^{(k)}\|_{\infty}\leq\delta$.

For the sparsity of $\overline{v}^{(k+1)}-\overline{v}^{(k)}$, we
first bound the size of the set $I_{\ell}\defeq\{i:|v_{i}^{(k)}-v_{i}^{(k-2^{\ell})}|\geq\delta/(2\left\lceil \log n\right\rceil )\}$.
Note that
\[
|I_{\ell}|\cdot\frac{\delta^{2}}{5\log^{2}n}\leq\sum_{i=1}^{n}|v_{i}^{(k)}-v_{i}^{(k-2^{\ell})}|^{2}\leq2^{\ell}\sum_{i=1}^{n}\sum_{t=k-2^{\ell}}^{k-1}|v_{i}^{(t+1)}-v_{i}^{(t)}|^{2}\leq2^{2\ell}\beta^{2}
\]
where we used $\|v^{(t+1)}-v^{(t)}\|_{2}\leq\beta$ at the end. Hence,
we have $|I_{\ell}|=O(2^{2\ell}(\beta/\delta)^{2}\log^{2}n).$ Hence,
the total number of changes is bounded by 
\[
|I|\leq\sum_{\ell=0}^{\ell_{k}}|I_{\ell}|=O(2^{2\ell_{k}}(\beta/\delta)^{2}\log^{2}n).
\]
\end{proof}

\subsection{Inverse Maintenance\label{subsec:IPM_Inverse_Maintenance}}

In this section, we discuss how to maintain the solution of the linear
system (Newton step) efficiently. Although both the matrix and the
vector of the Newton step changes during the algorithm, we can simplify
by moving the vector inside the matrix.
\begin{fact}
\label{fact:block-inverse}For any invertible matrix $\mm\in\R^{n\times n}$
and any vector $v\in\Rn$, we have
\[
\left[\begin{array}{cc}
\mm & v\\
0 & -1
\end{array}\right]^{-1}=\left[\begin{array}{cc}
\mm^{-1} & \mm^{-1}v\\
0 & -1
\end{array}\right].
\]
\end{fact}

Hence, the question of maintaining the solution reduces to the problem
of maintaining a column of the inverse of the matrix
\begin{equation}
\mm_{\ox,\os,\overline{r}}\defeq\left(\begin{array}{cccc}
\oms & \omx & \mzero & \nabla\Phi(\overline{r})\\
\ma & \mzero & \mzero & 0\\
\mzero & \mi & \ma^{\top} & 0\\
0 & 0 & 0 & -1
\end{array}\right).\label{eq:M_xsr_def}
\end{equation}
When we update $\ox,\os,\overline{r}$ to $\ox+\delta_{\ox},\os+\delta_{\os},\overline{r}+\delta_{\overline{r}}$
with $q$ coordinates modified in total, there are only $q$ columns
in $\mm$ that change. Hence, we can compute the update of the inverse
of $\mm$ using the Woodbury matrix identity (Equation (\ref{eq:Woodbury})).
The idea of using the Woodbury identity together with fast matrix
multiplication goes back to Vaidya \cite{vaidya1989speeding}.
\begin{lem}
\label{lem:maintain_inverse}Given vectors $\ox,\os,\overline{r}\in\R^{n}$
and the update $\delta_{\ox},\delta_{\os},\delta_{\overline{r}}\in\R^{n}$.
Let $q=\|\delta_{\ox}\|_{0}+\|\delta_{\os}\|_{0}+\|\delta_{\overline{r}}\|_{0}$
and $T_{m,n,\ell}$ be the cost of multiplying an $m\times n$ matrix
with an $n\times\ell$ matrix. Then,
\begin{itemize}
\item Given $\mm_{\ox,\os,\overline{r}}^{-1}$, we can compute $\mm_{\ox+\delta_{\ox},\os+\delta_{\os},\overline{r}+\delta_{\overline{r}}}^{-1}$
in time $O(T_{n,q,n})$.
\item Given $\mm_{\ox,\os,\overline{r}}^{-1}$ and $\mm_{\ox,\os,\overline{r}}^{-1}b$,
we can compute $\mm_{\ox+\delta_{\ox},\os+\delta_{\os},\overline{r}+\delta_{\overline{r}}}^{-1}b$
in time $O(T_{q,q,q}+nq)$.
\end{itemize}
\end{lem}

\begin{proof}
We write $\mm_{0}=\mm_{\ox,\os,\overline{r}}$ and $\mm_{1}=\mm_{\ox+\delta_{\ox},\os+\delta_{\os},\overline{r}+\delta_{\overline{r}}}$.
Note that $\mm_{1}$ and $\mm_{0}$ are off by just $q$ entries.
Hence, we can write
\[
\mm_{1}=\mm_{0}+\mathbf{U}\mathbf{C}\mathbf{V}
\]
where $\mathbf{U}$ consists of $q$ columns of identity matrix, $\mathbf{C}$
is a $q\times q$ matrix and $\mathbf{V}$ consists of $q$ rows of
identity matrix. Hence, the Woodbury matrix identity shows that
\begin{equation}
\mm_{1}^{-1}=(\mm_{0}+\mathbf{U}\mathbf{C}\mathbf{V})^{-1}=\mm_{0}^{-1}-\mm_{0}^{-1}\mathbf{U}(\mathbf{C}^{-1}+\mathbf{V}\mm_{0}^{-1}\mathbf{U})^{-1}\mathbf{V}\mm_{0}^{-1}.\label{eq:Woodbury}
\end{equation}
Note that $\mm_{0}^{-1}\mathbf{U}$, $\mathbf{V}\mm_{0}^{-1}\mathbf{U}$,
$\mathbf{V}\mm_{0}^{-1}$ are just blocks of $\mm_{0}^{-1}$ and no
computation is needed. Hence, we can compute $(\mathbf{C}^{-1}+\mathbf{V}\mm_{0}^{-1}\mathbf{U})^{-1}$
in the time to invert two $q\times q$ matrices, which is $O(T_{q,q,q})$.
The rest of the formula can be computed in $O(T_{n,q,q}+T_{n,q,n})=O(T_{n,q,n})$
time. In total, the runtime is $O(T_{n,q,n})$.

For computing $\mm_{1}^{-1}b$, we note that
\[
\mm_{1}^{-1}b=\mm_{0}^{-1}b-\mm_{0}^{-1}\mathbf{U}(\mathbf{C}^{-1}+\mathbf{V}\mm_{0}^{-1}\mathbf{U})^{-1}\mathbf{V}\mm_{0}^{-1}b.
\]
Since $\mm_{0}^{-1}b$ is given, the above formula can be computed
in $O(T_{q,q,q}+nq)$ time where the $O(nq)$ term comes from multiplying
a $q\times n$ matrix with an $n$-vector and a $n\times q$ matrix
with a $q$-vector.
\end{proof}
To use the previous lemma, we use the following estimate for $T_{n,r,n}$.
\begin{defn}
The exponent of matrix multiplication $\omega$ is the infimum among
all $\omega\geq0$ such that it takes $n^{\omega+o(1)}$ time to multiply
an $n\times n$ matrix by an $n\times n$ matrix. The dual exponent
of matrix multiplication $\alpha$ is the supremum among all $\alpha\geq0$
such that it takes $n^{2+o(1)}$ time to multiply an $n\times n$
matrix by an $n\times n^{\alpha}$ matrix. Currently, $\omega\leq2.3729$
\cite{coppersmith1987matrix,williams2012multiplying,le2014powers,alman2021refined}
and $\alpha\geq0.3138$ \cite{le2014powers,gall2018improved}.
\end{defn}

\begin{lem}
\label{lem:Tnnr_est}For $r\leq n$, we have $T_{n,r,n}=n^{2+o(1)}+n^{\omega-\frac{\omega-2}{1-\alpha}+o(1)}r^{\frac{\omega-2}{1-\alpha}}$.
\end{lem}

\begin{algorithm2e}[H]

\caption{$\mathtt{FastRobustStep}(\ma,x,s,t_{\mathrm{start}},t_{\mathrm{end}})$}

\label{alg:FastRobustStep}

\textbf{Assume $2^{2\ell_{*}}\leq n^{\alpha}$.}

\textbf{Define} $r=(xs-t)/t$ and $\Phi$ according to (\ref{eq:Phi_potential})
with $\lambda=16\log40n$.

\textbf{Invariant:} $(x,s)\in\PriR^{\circ}\times\DualR^{\circ}$ and
$\Phi(r)\leq16n$.

Let $t=t_{\mathrm{start}}$, $h=1/(128\lambda\sqrt{n})$ and $n$
be the number of columns in $\ma$.

Let $x^{(0)}=\ox^{(0)}=x,s^{(0)}=\os^{(0)}=s,r^{(0)}=\overline{r}^{(0)}=(xs-t)/t$.

Let $\mathbf{T}=\mm_{\ox^{(0)},\os^{(0)},\overline{r}^{(0)}}^{-1}$
(defined in (\ref{eq:M_xsr_def})) and $u=\mathbf{T}e_{2n+d+1}$.

\Repeat{$t\neq t_{\mathrm{end}}$}{

Let $t'=\max(t/(1+h),t_{\mathrm{end}})$, $\overline{\delta}_{\mu}=-\frac{t'}{32\lambda}\frac{\overline{g}}{\|\overline{g}\|_{2}}$,
$\overline{g}=\nabla\Phi(\overline{r})$.

\uIf{$k=0\mod2^{\ell_{*}}$}{

Update $\mathbf{T}$ to $\mm_{\ox^{(k)},\os^{(k)},\overline{r}^{(k)}}^{-1}$
using Lemma \ref{lem:maintain_inverse}.

$u\leftarrow\mathbf{T}e_{2n+d+1}$, $v\leftarrow u$.

}\Else{

Update $v$ to $\mm_{\ox^{(k)},\os^{(k)},\overline{r}^{(k)}}^{-1}e_{2n+d+1}$
using vector $u$ and Lemma \ref{lem:maintain_inverse}.

}

Let $(\delta_{x},\delta_{s})$ be the first $2n$ coordinates of $v$.

Let $x^{(k+1)}=x^{(k)}+\delta_{x}$, $s^{(k+1)}=s^{(k)}+\delta_{s}$
and $t\leftarrow t'$.

Let $r^{(k+1)}=(x^{(k+1)}s^{(k+1)}-t)/t$.

$\ln\overline{x}^{(k+1)}=\mathtt{SelectVector}(\ln\overline{x}^{(k)},\ln x^{(0)},\ln x^{(1)},\cdots,\ln x^{(k+1)},1/48)$.

$\ln\overline{s}^{(k+1)}=\mathtt{SelectVector}(\ln\overline{s}^{(k)},\ln s^{(0)},\ln s^{(1)},\cdots,\ln s^{(k+1)},1/48)$.

$\overline{r}^{(k+1)}=\mathtt{SelectVector}(\overline{r}^{(k)},r^{(0)},r^{(1)},\cdots,r^{(k+1)},1/(48\lambda))$.

Set $k\leftarrow k+1$.

}

\textbf{Return} $(x,s)$

\end{algorithm2e}

We note that in the algorithm above, we do not compute the inverse
of $\mm$ in every iteration, which would be too expensive. Rather,
we compute $\mm^{-1}b$ as needed by using the previously computed
$\mm^{-1}$ (from possibly many iterations ago) with a low-rank update
using the Woodbury formula that we maintain.

Combining the Lemma \ref{lem:maintain_inverse} with Lemma \ref{lem:SelectVector},
we have the following guarantee.
\begin{lem}
\label{lem:FastRobustStepTime}Setting $2^{2\ell_{*}}=\min(n^{\alpha},n^{2/3})$,
$\mathtt{FastRobustStep}$ (Algorithm \ref{alg:FastRobustStep}) takes
time 
\[
O((n^{\omega+o(1)}+n^{2+1/6+o(1)}+n^{5/2-\alpha/2+o(1)})\log(t_{\mathrm{end}}/t_{\mathrm{start}})).
\]
\end{lem}

\begin{proof}
The bottleneck of $\mathtt{FastRobustStep}$ is the time to update
$v$ and $\mathbf{T}$. This depends on the number of coordinates
updated in $\ox,\os,\overline{r}$. Lemma \ref{lem:change_xsr} shows
that $\ln x,\ln s$ and $r$ change by at most $\alpha\defeq1/(8\lambda)$
in $\ell_{2}$ norm per step. Since we set the error of $\mathtt{SelectVector}$
to be $\delta\defeq1/(48\lambda)$ (or larger), Lemma \ref{lem:SelectVector}
shows that $q_{k}\defeq O(2^{2\ell_{k}}\log^{2}n)$ coordinates in
$\ox,\os,\overline{r}$ are updated at the $k$-th step where $\ell_{k}$
is the largest integer $\ell$ with $k=0\mod2^{\ell}$. We can now
bound all the computation costs as follows.

\textbf{Cost of updating $v$}: We update $u$ whenever $k=0\mod2^{\ell_{*}}$.
Within that $2^{\ell_{*}}$ steps, the number of coordinates updated
in $\ox,\os,\overline{r}$ is bounded by
\[
q\defeq\sum_{k=1}^{2^{\ell_{*}}-1}q_{k}=\sum_{k=1}^{2^{\ell_{*}}-1}O(2^{2\ell_{k}}\log^{2}n)=O(2^{2\ell_{*}}\log^{2}n).
\]
Therefore, $\mm_{\ox^{(k)},\os^{(k)},\overline{r}^{(k)}}$ and $\mathbf{T}^{-1}$
are off by at most $q$ coordinates. Lemma \ref{lem:maintain_inverse}
shows that it takes 
\[
O(T_{q,q,q}+nq)=\widetilde{\widetilde{O}}(2^{2\ell_{*}\omega}+n2^{2\ell_{*}})
\]
time to compute $v=\mm_{\ox^{(k)},\os^{(k)},\overline{r}^{(k)}}^{-1}e_{2n+d+1}$
using $u=\mathbf{T}e_{2n+d+1}$, where we used $\widetilde{\widetilde{O}}$
to omit $n^{o(1)}$ terms.

\textbf{Cost of updating $\mathbf{T}$}: For the $k$-th step that
updates $\mathbf{T}$, the number of coordinates updated in $\ox,\os,\overline{r}$
is bounded by $q+O(2^{2\ell_{k}}\log^{2}n)=O(2^{2\ell_{k}}\log^{2}n)$
where the first term is due to the delayed updates and the second
term is due to the updates at that step. Lemma \ref{lem:maintain_inverse}
shows that it takes $\widetilde{\widetilde{O}}(T_{n,n,2^{2\ell_{k}}})$
time to update $\mathbf{T}.$ Since $2^{2\ell_{k}}$ updates happen
every $2^{\ell_{k}}$ iterations, the amortized cost is
\begin{align*}
\widetilde{\widetilde{O}}(\sum_{\ell=\ell_{*}}^{\frac{1}{2}\log n}2^{-\ell}T_{n,n,2^{2\ell}}) & =\widetilde{\widetilde{O}}(\sum_{\ell=\ell_{*}}^{\frac{1}{2}\log n}(n^{\omega-\frac{\omega-2}{1-\alpha}}2^{2\ell\cdot\frac{\omega-2}{1-\alpha}-\ell}+n^{2}2^{-\ell}))\\
 & =\widetilde{\widetilde{O}}(\sum_{\ell=\ell_{*}}^{\frac{1}{2}\log n}n^{\omega-\frac{\omega-2}{1-\alpha}}2^{2\ell\cdot\frac{\omega-2}{1-\alpha}-\ell}+n^{2}2^{-\ell_{*}})
\end{align*}
where we used Lemma \ref{lem:Tnnr_est}. The sum above is dominated
by either the term at $\ell=\ell_{*}$ or the term at $\ell=\frac{1}{2}\log n$.
Hence, the amortized cost of updating $\mathbf{T}$ is
\[
\widetilde{\widetilde{O}}(n^{\omega-\frac{\omega-2}{1-\alpha}}2^{2\ell_{*}\cdot\frac{\omega-2}{1-\alpha}-\ell_{*}}+n^{\omega-\frac{1}{2}}+n^{2}2^{-\ell_{*}})=\widetilde{\widetilde{O}}(n^{\omega-\frac{1}{2}}+n^{2}2^{-\ell_{*}})
\]
where we used $n^{\omega-\frac{\omega-2}{1-\alpha}}2^{2\ell_{*}\cdot\frac{\omega-2}{1-\alpha}}\leq n^{2}$
since $2^{2\ell_{*}}\leq n^{\alpha}$.

\textbf{Cost of initializing $\mathbf{T}$ and $u$}: $\widetilde{\widetilde{O}}(n^{\omega})$.

Since there are $\sqrt{n}\log(t_{\mathrm{end}}/t_{\mathrm{start}})$
steps, the total cost is
\begin{align*}
 & \widetilde{\widetilde{O}}(n^{\omega}+\sqrt{n}\log(t_{\mathrm{end}}/t_{\mathrm{start}})(2^{2\ell_{*}\omega}+n2^{2\ell_{*}}+n^{\omega-\frac{1}{2}}+n^{2}2^{-\ell_{*}}))\\
= & \widetilde{\widetilde{O}}(\sqrt{n}\log(t_{\mathrm{end}}/t_{\mathrm{start}})(n2^{2\ell_{*}}+n^{\omega-\frac{1}{2}}+n^{2}2^{-\ell_{*}}))
\end{align*}
where we used $2^{2\ell_{*}\omega}\leq n^{\alpha\omega}\leq n$. Putting
$2^{2\ell_{*}}=\min(n^{\alpha},n^{2/3})$, we have
\[
\widetilde{\widetilde{O}}((n^{\omega}+n^{2+1/6}+n^{5/2-\alpha/2})\log(t_{\mathrm{end}}/t_{\mathrm{start}})).
\]
\end{proof}
Following Section \ref{subsec:SolveLP_apx}, we can find the initial
point by modifying the linear program and this gives the following
theorem.
\begin{thm}
\label{thm:RobustLPApproximate}Consider a linear program $\min_{\ma x=b,x\geq0}c^{\top}x$
with $n$ variables and $d$ constraints. Assume the linear program
has inner radius $r$, outer radius $R$ and Lipschitz constant $L$
(see Definition \ref{def:LP_para}), we can find $x$ such that
\begin{align*}
c^{\top}x & \leq\min_{\ma x=b,x\geq0}c^{\top}x+\delta LR,\\
\ma x & =b,\\
x & \geq0.
\end{align*}
in time
\[
O((n^{\omega+o(1)}+n^{2+1/6+o(1)}+n^{5/2-\alpha/2+o(1)})\log(R/(\delta r))).
\]

If we further assume that the solution $x^{*}=\arg\min_{\ma x=b,x\geq0}c^{\top}x$
is unique and that $c^{\top}x\geq c^{\top}x^{*}+\eta LR$ for any
other vertex $x$ of $\{\ma x=b,x\geq0\}$ for some $\eta>\delta\geq0$,
then we have that $\|x-x^{*}\|_{2}\leq\frac{2\delta R}{\eta}$.
\end{thm}

\begin{proof}
The algorithm for find $x$ is the same as $\mathtt{SlowSolveLP}$
except that the function $\mathtt{L2Step}$ is replaced by the function
$\mathtt{FastRobustStep}$. The runtime of $\mathtt{FastRobustStep}$
is analyzed in Lemma \ref{lem:FastRobustStepTime}. Since $\mathtt{FastRobustStep}$
is a instantiation of $\mathtt{RobustStep}$, its output is analyzed
in Lemma \ref{lem:RobustStep}.
\end{proof}

\paragraph{Historical Note.}

The interior-point method was pioneered by Karmarkar \cite{Karmarkar84}
and developed in beautiful ways (including \cite{renegar1988polynomial,vaidya1989speeding,nesterov1994interior,nesterov2002riemannian,lee2014path,lee2015faster}).
This classical approach appeared to reach its limit of requiring the
solution of $\sqrt{n}$ linear systems until the paper of Cohen, Lee
and Song \cite{cohen2018solving} which reduced the complexity of
approximate linear programming to $n^{\omega+o(1)}\log(1/\epsilon)$,
by introducing the robust central path method. Using further insights,
their algorithm was derandomized by van den Brand \cite{DBLP:conf/soda/Brand20}.
The technique has been extended in subsequent papers to other optimization
problems \cite{lee2019solving,jiang2020faster,BrandLSS20,jiang2021faster,van2021unifying,huang2021solving},
and has also played a crucial role in faster algorithms for classical
combinatorial optimization problems such as matchings and flows \cite{BrandLNPSS20,BrandLLSSSW21,treeLP}.

\bibliographystyle{plain}
\bibliography{ipm,main}

\begin{thebibliography}{10}

\bibitem{alman2021refined}
Josh Alman and Virginia~Vassilevska Williams.
\newblock A refined laser method and faster matrix multiplication.
\newblock In {\em Proceedings of the 2021 ACM-SIAM Symposium on Discrete
  Algorithms (SODA)}, pages 522--539. SIAM, 2021.

\bibitem{cohen2018solving}
Michael~B Cohen, Yin~Tat Lee, and Zhao Song.
\newblock Solving linear programs in the current matrix multiplication time.
\newblock {\em arXiv preprint arXiv:1810.07896}, 2018.

\bibitem{coppersmith1987matrix}
Don Coppersmith and Shmuel Winograd.
\newblock Matrix multiplication via arithmetic progressions.
\newblock In {\em Proceedings of the nineteenth annual ACM symposium on Theory
  of computing}, pages 1--6, 1987.

\bibitem{treeLP}
Sally Dong, Yin~Tat Lee, and Guanghao Ye.
\newblock A nearly-linear time algorithm for linear programs with small
  treewidth: A multiscale representation of robust central path.
\newblock In {\em Proceedings of the 53rd Annual ACM SIGACT Symposium on Theory
  of Computing}, STOC 2021, pages 1784--1797, New York, NY, USA, 2021.
  Association for Computing Machinery.

\bibitem{gall2018improved}
Fran{\c{c}}ois~Le Gall and Florent Urrutia.
\newblock Improved rectangular matrix multiplication using powers of the
  coppersmith-winograd tensor.
\newblock In {\em Proceedings of the Twenty-Ninth Annual ACM-SIAM Symposium on
  Discrete Algorithms}, pages 1029--1046. SIAM, 2018.

\bibitem{huang2021solving}
Baihe Huang, Shunhua Jiang, Zhao Song, and Runzhou Tao.
\newblock Solving tall dense sdps in the current matrix multiplication time.
\newblock {\em arXiv preprint arXiv:2101.08208}, 2021.

\bibitem{jiang2020faster}
Haotian Jiang, Tarun Kathuria, Yin~Tat Lee, Swati Padmanabhan, and Zhao Song.
\newblock A faster interior point method for semidefinite programming.
\newblock In {\em 2020 IEEE 61st Annual Symposium on Foundations of Computer
  Science (FOCS)}, pages 910--918. IEEE, 2020.

\bibitem{jiang2021faster}
Shunhua Jiang, Zhao Song, Omri Weinstein, and Hengjie Zhang.
\newblock A faster algorithm for solving general lps.
\newblock In {\em Proceedings of the 53rd Annual ACM SIGACT Symposium on Theory
  of Computing}, pages 823--832, 2021.

\bibitem{Karmarkar84}
N.~Karmarkar.
\newblock A new polynomial-time algorithm for linear programming.
\newblock {\em Combinatorica}, 4(4):373--396, 1984.

\bibitem{klivans2001randomness}
Adam~R Klivans and Daniel Spielman.
\newblock Randomness efficient identity testing of multivariate polynomials.
\newblock In {\em Proceedings of the thirty-third annual ACM symposium on
  Theory of computing}, pages 216--223, 2001.

\bibitem{le2014powers}
Fran{\c{c}}ois Le~Gall.
\newblock Powers of tensors and fast matrix multiplication.
\newblock In {\em Proceedings of the 39th international symposium on symbolic
  and algebraic computation}, pages 296--303, 2014.

\bibitem{lee2014path}
Yin~Tat Lee and Aaron Sidford.
\newblock Path finding methods for linear programming: Solving linear programs
  in {\~o} (vrank) iterations and faster algorithms for maximum flow.
\newblock In {\em Foundations of Computer Science (FOCS), 2014 IEEE
  55\textsuperscript{th} Annual Symposium on}, pages 424--433. IEEE, 2014.

\bibitem{lee2015faster}
Yin~Tat Lee, Aaron Sidford, and Sam Chiu-wai Wong.
\newblock A faster cutting plane method and its implications for combinatorial
  and convex optimization.
\newblock In {\em Foundations of Computer Science (FOCS), 2015 IEEE 56th Annual
  Symposium on}, pages 1049--1065. IEEE, 2015.

\bibitem{lee2019solving}
Yin~Tat Lee, Zhao Song, and Qiuyi Zhang.
\newblock Solving empirical risk minimization in the current matrix
  multiplication time.
\newblock In {\em Conference on Learning Theory}, pages 2140--2157. PMLR, 2019.

\bibitem{nesterov1994interior}
Yurii Nesterov and Arkadii Nemirovskii.
\newblock {\em Interior-point polynomial algorithms in convex programming}.
\newblock SIAM, 1994.

\bibitem{nesterov2002riemannian}
Yurii~E Nesterov, Michael~J Todd, et~al.
\newblock On the riemannian geometry defined by self-concordant barriers and
  interior-point methods.
\newblock {\em Foundations of Computational Mathematics}, 2(4):333--361, 2002.

\bibitem{renegar1988polynomial}
James Renegar.
\newblock A polynomial-time algorithm, based on newton's method, for linear
  programming.
\newblock {\em Mathematical Programming}, 40(1):59--93, 1988.

\bibitem{vaidya1989speeding}
Pravin~M. Vaidya.
\newblock Speeding-up linear programming using fast matrix multiplication
  (extended abstract).
\newblock In {\em 30th Annual Symposium on Foundations of Computer Science,
  Research Triangle Park, North Carolina, USA, 30 October - 1 November 1989},
  pages 332--337, 1989.

\bibitem{DBLP:conf/soda/Brand20}
Jan van~den Brand.
\newblock A deterministic linear program solver in current matrix
  multiplication time.
\newblock In Shuchi Chawla, editor, {\em Proceedings of the 2020 {ACM-SIAM}
  Symposium on Discrete Algorithms, {SODA} 2020, Salt Lake City, UT, USA,
  January 5-8, 2020}, pages 259--278. {SIAM}, 2020.

\bibitem{van2021unifying}
Jan van~den Brand.
\newblock Unifying matrix data structures: Simplifying and speeding up
  iterative algorithms.
\newblock In {\em Symposium on Simplicity in Algorithms (SOSA)}, pages 1--13.
  SIAM, 2021.

\bibitem{BrandLLSSSW21}
Jan van~den Brand, Yin~Tat Lee, Yang~P Liu, Thatchaphol Saranurak, Aaron
  Sidford, Zhao Song, and Di~Wang.
\newblock Minimum cost flows, {MDP}s, and l1-regression in nearly linear time
  for dense instances.
\newblock In {\em Proceedings of the 53rd Annual ACM SIGACT Symposium on Theory
  of Computing}, pages 859--869, 2021.

\bibitem{BrandLNPSS20}
Jan van~den Brand, Yin-Tat Lee, Danupon Nanongkai, Richard Peng, Thatchaphol
  Saranurak, Aaron Sidford, Zhao Song, and Di~Wang.
\newblock Bipartite matching in nearly-linear time on moderately dense graphs.
\newblock In {\em 2020 IEEE 61st Annual Symposium on Foundations of Computer
  Science (FOCS)}, pages 919--930. IEEE, 2020.

\bibitem{BrandLSS20}
Jan van~den Brand, Yin~Tat Lee, Aaron Sidford, and Zhao Song.
\newblock Solving tall dense linear programs in nearly linear time.
\newblock In {\em Proceedings of the 52nd Annual ACM SIGACT Symposium on Theory
  of Computing}, pages 775--788, 2020.

\bibitem{williams2012multiplying}
Virginia~Vassilevska Williams.
\newblock Multiplying matrices faster than coppersmith-winograd.
\newblock In {\em Proceedings of the forty-fourth annual ACM symposium on
  Theory of computing}, pages 887--898, 2012.

\end{thebibliography}

\appendix

\section{Finding a Point on the Central Path\label{sec:IPM_interior}}

We continue from the discussion in Section \ref{subsec:SolveLP_apx}.

First, we show that $x^{(0)}$ defined in Theorem \ref{thm:IPM_interior}
is indeed on the central path of the modified linear program.
\begin{lem}
\label{lem:ipm_interior_modified}The modified linear program (Definition
\ref{def:IPM_interior_modified}) has an explicit central path point
$x^{(0)}=(x_{c}^{+},x_{c}^{-},\overline{R})$ at $t$.
\end{lem}

\begin{proof}
Recall that we say $(x^{+},x^{-},x^{\theta})$ is on the central path
at $t$ if $x^{+},x^{-},x^{\theta}$ are positive and it satisfies
the following equation
\begin{align}
\ma x^{+}-\ma x^{-} & =b,\nonumber \\
\sum_{i=1}^{n}x_{i}^{+}+x^{\theta} & =\widetilde{b},\nonumber \\
\ma^{\top}y+\lambda+s^{+} & =c,\label{eq:KKT_modified_LP}\\
-\ma^{\top}y+s^{-} & =\widetilde{c},\nonumber \\
\lambda+s^{\theta} & =0,\nonumber 
\end{align}
for some $s^{+},s^{-}\in\R_{>0}^{n}$, $s^{\theta}>0$, $y\in\R^{d}$
and $\lambda\in\R$.

Now, we verify the solution $x^{+}=\frac{t}{c+t/\overline{R}}$, $x^{-}=\frac{t}{c+t/\overline{R}}-x_{\circ}$,
$x^{\theta}=\overline{R}$, $x_{\circ}=\ma^{\top}(\ma\ma^{\top})^{-1}b$,
$y=0$, $s^{+}=\frac{t}{x^{+}}$, $s^{-}=\frac{t}{x^{-}}$, $s^{\theta}=\frac{t}{x^{\theta}}$,
$\lambda=-s^{\theta}$. Using $\ma x_{\circ}=b$, one can check it
satisfies all the equality constraints above. 

For the inequality constraints, using $\|c\|_{\infty}\leq L$ and
$t\geq8L\overline{R}$, we have
\begin{equation}
\frac{3}{4}\overline{R}\leq\frac{t}{L+t/\overline{R}}\leq x_{i}^{+}\leq\frac{t}{-L+t/\overline{R}}\leq\frac{3}{2}\overline{R}\label{eq:x_pos_bound}
\end{equation}
and hence $x^{+}>0$ and so is $s^{+}$. Since $\|x_{\circ}\|_{2}\leq R\leq\frac{\overline{R}}{2}$
and $x_{i}^{+}\geq\frac{3}{4}\overline{R}$ for all $i$, we have
$x_{i}^{-}\geq0$ for all $i$. Hence, $x^{-}$ and $s^{-}$ are positive.
Finally, $x^{\theta}$ and $s^{\theta}$ are positive. This proves
that $(\frac{t}{c+t/\overline{R}},\frac{t}{c+t/\overline{R}}-x_{\circ},\overline{R})$
is on the central path point at $t$.
\end{proof}
Next, we show that the near-central-path point $(x,s)$ at $t=LR$
is far from the constraints $x^{+}\geq0$ and is close to the constraints
$x^{-}\geq0$. The proof for both involves the same idea: use the
optimality condition of $x$. Throughout the rest of the section,
we are given $(x,s)\in\mathcal{P}_{\overline{R},t}\times\mathcal{D}_{\overline{R},t}$
such that $\mu=xs$ satisfies
\[
\frac{5}{6}LR\leq\mu\leq\frac{7}{6}LR.
\]
We write $\mu$ into its three parts $(\mu^{+},\mu^{-},\mu^{\theta})$.
By Lemma \ref{lem:primal_dual_representation}, we have that $x\defeq(x^{+},x^{-},x^{\theta})$
minimizes the function
\[
f(x^{+},x^{-},x^{\theta})\defeq c^{\top}x^{+}+\widetilde{c}^{\top}x^{-}-\sum_{i=1}^{n}\mu_{i}^{+}\log x_{i}^{+}-\sum_{i=1}^{n}\mu_{i}^{-}\log x_{i}^{-}-\mu^{\theta}\log x^{\theta}
\]
over the domain $\mathcal{P}_{\overline{R},t}$. The gradient of $f$
is a bit complicated. We avoid it by considering the directional derivative
at $x$ along the direction ``$v-x$'' where $v$ is the point such
that $\ma v=b$ and $v\geq r$. Since our domain is in $\mathcal{P}_{\overline{R},t}\subset\R^{2n+1}$,
we need to lift $v$ to higher dimension. So we define the point
\begin{align*}
v^{-} & =\min(x^{-},\frac{8L\overline{R}}{t}\cdot R),\\
v^{+} & =v+v^{-},\\
v^{\theta} & =\widetilde{b}-\sum_{i=1}^{n}v_{i}^{+}.
\end{align*}
First, we need to get some basic bounds on $\widetilde{b}$ and $\widetilde{c}$.
\begin{lem}
\label{lem:ipm_propert_modified}We have that $\frac{3}{4}n\overline{R}\leq\widetilde{b}\leq3n\overline{R}$
and $\widetilde{c}_{i}\geq t/(2\overline{R})$ for all $i$.
\end{lem}

\begin{proof}
By (\ref{eq:x_pos_bound}), we have $\frac{3}{4}\overline{R}\leq x_{c,i}^{+}\leq\frac{3}{2}\overline{R}$.
By the definition of $\widetilde{b}$, we have
\[
\widetilde{b}=\sum_{i=1}^{n}x_{c,i}^{+}+\overline{R}\leq\frac{3}{2}n\overline{R}+\overline{R}\leq3n\overline{R}.
\]
Similarly, we have $\widetilde{b}=\sum_{i}x_{i}^{+}+\overline{R}\geq\frac{3}{4}n\overline{R}$. 

For the bound on $\widetilde{c}$, recall that $\widetilde{c}=t/x_{c}^{-}$
with $x_{c}^{-}=x_{c}^{+}-x_{\circ}$ and $x_{\circ}=\ma^{\top}(\ma\ma^{\top})^{-1}b=\arg\min_{\ma x=b}\|x\|_{2}$.
Since we assumed the linear program has outer radius $R$, we have
that $\|x_{\circ}\|_{2}\leq R$. Hence, 
\[
x_{c,i}^{-}\leq\frac{3}{2}\overline{R}+R\leq2\overline{R}.
\]
Therefore, $\widetilde{c}\geq t/(2\overline{R})$.
\end{proof}
The following lemma shows that $(v^{+},v^{-},v^{\theta})\in\mathcal{P}_{\overline{R},t}$.
\begin{lem}
\label{lem:IPM_exact_inner_pt}We have that $(v^{+},v^{-},v^{\theta})\in\mathcal{P}_{\overline{R},t}$.
Furthermore, we have $v^{\theta}\geq\frac{1}{2}n\overline{R}$.
\end{lem}

\begin{proof}
Note that $(v^{+},v^{-},v^{\theta})$ satisfies the linear constraints
of $\mathcal{P}_{\overline{R},t}$ by construction. It suffices to
prove the vector is positive. Since $x^{-}>0$, we have $v^{-}>0$.
Since $v\geq r$, we also have $v^{+}>0$. For $v^{\theta}$, we use
$\widetilde{b}\geq\frac{3}{4}n\overline{R}$ (Lemma \ref{lem:ipm_propert_modified}),
$v\leq R$ and $v^{-}\leq\frac{8L\overline{R}}{t}\cdot R\leq R$ to
get
\begin{align*}
v^{\theta}=\widetilde{b}-\sum_{i=1}^{n}v_{i}^{+} & \geq\frac{3}{4}n\overline{R}-\sum_{i=1}^{n}(v_{i}+v_{i}^{-})\geq\frac{1}{2}n\overline{R}.
\end{align*}
\end{proof}
Next, we define the path $p(t)=(1-t)(x^{+},x^{-},x^{\theta})+t(v^{+},v^{-},v^{\theta})$.
Since $p(0)$ minimizes $f$, we have that $\frac{d}{dt}f(p(t))|_{t=0}\geq0$.
In particular, we have
\begin{align}
0 & \leq\frac{d}{dt}f(p(t))|_{t=0}\nonumber \\
 & =c^{\top}(v^{+}-x^{+})+\widetilde{c}^{\top}(v^{-}-x^{-})-\sum_{i=1}^{n}\frac{\mu_{i}^{+}}{x_{i}^{+}}(v^{+}-x^{+})_{i}-\sum_{i=1}^{n}\frac{\mu_{i}^{-}}{x_{i}^{-}}(v^{-}-x^{-})_{i}-\frac{\mu^{\theta}}{x^{\theta}}(v_{\theta}-x^{\theta})\nonumber \\
 & =\frac{\mu^{\theta}}{x^{\theta}}(x^{\theta}-v_{\theta})+\sum_{i=1}^{n}(c_{i}-\frac{\mu_{i}^{+}}{x_{i}^{+}})(v^{+}-x^{+})_{i}+\sum_{i=1}^{n}(\widetilde{c}_{i}-\frac{\mu_{i}^{-}}{x_{i}^{-}})(v^{-}-x^{-})_{i}.\label{eq:LP_exact_3_term}
\end{align}
Now, we bound each term one by one. For the first term, we note that
\begin{equation}
\frac{\mu^{\theta}}{x^{\theta}}(x^{\theta}-v_{\theta})\leq\mu^{\theta}\leq2LR.\label{eq:LP_exact_3_term_1}
\end{equation}
For the second term in (\ref{eq:LP_exact_3_term}), we have the following
\begin{lem}
\label{lem:LP_exact_3_term_1}We have that $\sum_{i=1}^{n}(c_{i}-\frac{\mu_{i}^{+}}{x_{i}^{+}})(v^{+}-x^{+})_{i}\leq4nL\overline{R}-\frac{LRr}{2\min_{i}x_{i}^{+}}$.
\end{lem}

\begin{proof}
Note that
\begin{align*}
\sum_{i=1}^{n}(c_{i}-\frac{\mu_{i}^{+}}{x_{i}^{+}})(v^{+}-x^{+})_{i} & =\sum_{i=1}^{n}(c_{i}v_{i}^{+}-\frac{\mu_{i}^{+}}{x_{i}^{+}}v_{i}^{+}-c_{i}x_{i}^{+}+\mu_{i}^{+})\\
 & \leq\sum_{i=1}^{n}c_{i}v_{i}^{+}+\sum_{i=1}^{n}\mu_{i}^{+}-\sum_{i=1}^{n}\frac{\mu_{i}^{+}}{x_{i}^{+}}v_{i}^{+}\\
 & \leq\|c\|_{\infty}\|v^{+}\|_{1}+2nLR-\frac{1}{2}\sum_{i=1}^{n}\frac{LRr}{x_{i}^{+}}
\end{align*}
where we used $\mu_{i}^{+}\in[\frac{LR}{2},2LR]$ and $v_{i}^{+}\geq v_{i}\geq r$
at the end. The result follows from $\|c\|_{\infty}\leq L$, $\|v^{+}\|_{1}\leq\widetilde{b}\leq3n\overline{R}$
(Lemma \ref{lem:ipm_propert_modified}).
\end{proof}
For the third term in (\ref{eq:LP_exact_3_term}), we have the following
\begin{lem}
\label{lem:LP_exact_3_term_2}We have that $\sum_{i=1}^{n}(\widetilde{c}_{i}-\frac{\mu_{i}^{-}}{x_{i}^{-}})(v^{-}-x^{-})_{i}\leq2LR-\frac{t}{4\overline{R}}\max_{i}x_{i}^{-}$.
\end{lem}

\begin{proof}
Using $v^{-}=\min(x^{-},\frac{8L\overline{R}}{t}\cdot R)$, we have
$v_{i}^{-}\leq x_{i}^{-}$. We can ignore the terms with $v_{i}^{-}=x_{i}^{-}$. 

For $v_{i}^{-}<x_{i}^{-}$, we have $x_{i}^{-}\geq\frac{8L\overline{R}}{t}R$.
Using $\widetilde{c}_{i}\geq\frac{t}{2\overline{R}}$ (Lemma \ref{lem:ipm_propert_modified}),
we have
\[
\widetilde{c}_{i}-\frac{\mu_{i}^{-}}{x_{i}^{-}}\geq\widetilde{c}_{i}-\frac{\mu_{i}^{-}}{\frac{8L\overline{R}}{t}R}\geq\widetilde{c}_{i}-\frac{2LR}{\frac{8L\overline{R}}{t}R}=\widetilde{c}_{i}-\frac{t}{4\overline{R}}\geq\frac{t}{4\overline{R}}.
\]
Hence, we have
\[
\sum_{i=1}^{n}(\widetilde{c}_{i}-\frac{\mu_{i}^{-}}{x_{i}^{-}})(v^{-}-x^{-})_{i}\leq\frac{t}{4\overline{R}}\sum_{i=1}^{n}(v^{-}-x^{-})_{i}\leq\frac{t}{4\overline{R}}(\frac{8L\overline{R}}{t}\cdot R-\max_{i}x_{i}^{-}).
\]
\end{proof}
Combining (\ref{eq:LP_exact_3_term}), (\ref{eq:LP_exact_3_term_1}),
Lemma \ref{lem:LP_exact_3_term_1} and Lemma \ref{lem:LP_exact_3_term_2},
we have
\begin{align*}
0 & \leq\frac{\mu^{\theta}}{x^{\theta}}(x^{\theta}-v_{\theta})+\sum_{i=1}^{n}(c_{i}-\frac{\mu_{i}^{+}}{x_{i}^{+}})(v^{+}-x^{+})_{i}+\sum_{i=1}^{n}(\widetilde{c}_{i}-\frac{\mu_{i}^{-}}{x_{i}^{-}})(v^{-}-x^{-})_{i}\\
 & \leq2LR+4nL\overline{R}-\frac{LRr}{2\min_{i}x_{i}^{+}}+2LR-\frac{t}{4\overline{R}}\max_{i}x_{i}^{-}\\
 & =5nL\overline{R}-\frac{LRr}{2\min_{i}x_{i}^{+}}-\frac{t}{4\overline{R}}\max_{i}x_{i}^{-}.
\end{align*}
Hence, we show that $(x^{+},x^{-},x^{\theta})$ is close to the central
path at $t=LR$ implies that $\min_{i}x_{i}^{+}$ cannot be too small
and $\max_{i}x_{i}^{-}$ cannot be too large
\[
\frac{LRr}{2\min_{i}x_{i}^{+}}+\frac{t}{4\overline{R}}\max_{i}x_{i}^{-}\leq5nL\overline{R}.
\]
In particular, this shows the following:
\begin{lem}
\label{lem:LP_exact_distance_OPT}We have that $\min_{i}x_{i}^{+}\geq\frac{Rr}{10n\overline{R}}$
and $\max_{i}x_{i}^{-}\leq\frac{20nL\overline{R}}{t}\cdot\overline{R}$.
\end{lem}

Now, we are ready to prove the second conclusion of Theorem \ref{thm:IPM_interior}.
\begin{lem}
\label{lemma:IPM_center_original}For any primal $x\defeq(x^{+},x^{-},x^{\theta})\in\mathcal{P}_{\overline{R},t}$
and dual $s\defeq(s^{+},s^{-},s^{\theta})\in\mathcal{D}_{\overline{R},t}$
such that $\frac{5}{6}LR\leq x_{i}s_{i}\leq\frac{7}{6}LR$, we have
that
\[
(x^{+}-x^{-},s^{+}-s^{\theta})\in\PriR\times\DualR
\]
and that $x_{i}^{-}\leq\epsilon x_{i}^{+}$ and $s^{\theta}\leq\epsilon s_{i}^{+}$
for all $i$.
\end{lem}

\begin{proof}
First we check $x\defeq x^{+}-x^{-}\in\PriR$. By the choice of $\overline{R}$
and $t$, Lemma \ref{lem:LP_exact_distance_OPT} shows that 
\[
\frac{\max_{i}x_{i}^{-}}{\min_{i}x_{i}^{+}}\leq\frac{\frac{20nL\overline{R}}{t}\cdot\overline{R}}{\frac{Rr}{10n\overline{R}}}=\frac{200n^{2}L\frac{\overline{R}^{3}}{Rr}}{t}\leq\epsilon.
\]
Hence, we have $x^{+}-x^{-}>0$ and that $\ma(x^{+}-x^{-})=b$. 

Next, we check $s\defeq s^{+}-s^{\theta}\in\DualR$. Since $x\in\mathcal{P}$,
we have $x\leq R$ and $x_{i}^{+}\leq\frac{3}{2}x_{i}\leq\frac{3}{2}R$.
Since $x_{i}^{+}s_{i}^{+}\geq\frac{5}{6}LR$, we have $s_{i}^{+}\geq\frac{1}{2}L.$
On the other hand, we have $x^{\theta}=\widetilde{b}-\sum_{i=1}^{n}x_{i}^{+}\geq\widetilde{b}-2nR\geq\frac{1}{2}n\overline{R}$
(Lemma \ref{lem:ipm_propert_modified}). Hence, $s^{\theta}\leq\frac{\frac{7}{6}LR}{\frac{1}{2}n\overline{R}}\leq\frac{5LR}{2n\overline{R}}.$
Combining both and the choice of $\overline{R}$, we have
\[
\frac{s^{\theta}}{\min_{i}s_{i}^{+}}\leq\frac{\frac{5LR}{2n\overline{R}}}{L/2}=\frac{5R}{n\overline{R}}\leq\epsilon
\]
Hence, we have $s=s^{+}-s^{\theta}>0$ and that $\ma^{\top}y+s=\ma^{\top}y+s^{+}-s^{\theta}=\ma^{\top}y+s^{+}+\lambda=c$
(See (\ref{eq:KKT_modified_LP})).
\end{proof}
To ensure the reduction does not increase the complexity of solving
linear system, we note that the linear constraint in the modified
linear program is
\[
\oma=\left[\begin{array}{ccc}
\ma & -\ma & 0\\
1 & 0 & 1
\end{array}\right]
\]
For any diagonal matrices $\mw_{1},\mw_{2}$ and any scalar $\alpha$,
we have
\[
\mh\defeq\oma\left[\begin{array}{ccc}
\mw_{1} & \mzero & 0\\
\mzero & \mw_{2} & 0\\
0 & 0 & \alpha
\end{array}\right]\oma^{\top}=\left[\begin{array}{cc}
\ma^{\top}(\mw_{1}+\mw_{2})\ma & \ma w_{1}\\
(\ma w_{1})^{\top} & \|w_{1}\|_{1}+\alpha
\end{array}\right].
\]
Note that the second row and second column block has size $1$. By
the block inverse formula (Fact \ref{fact:block-inverse}), $\mh^{-1}v$
is an explicit formula involving $(\ma^{\top}(\mw_{1}+\mw_{2})\ma)^{-1}v_{1:n}$
and $(\ma^{\top}(\mw_{1}+\mw_{2})\ma)^{-1}\ma w_{1}$. Hence, we can
compute $\mh^{-1}v$ by solving two linear systems of the form $\ma^{\top}\mw\ma$
with some extra linear work.
\end{document}